\documentclass[letterpaper]{article} %
\usepackage{aaai25}  %
\usepackage{times}  %
\usepackage{helvet}  %
\usepackage{courier}  %
\usepackage[hyphens]{url}  %
\usepackage{graphicx} %
\usepackage{natbib}  %
\usepackage{caption} %
\usepackage{algorithm}
\usepackage{algorithmic}

\usepackage{newfloat}
\usepackage{listings}
\DeclareCaptionStyle{ruled}{labelfont=normalfont,labelsep=colon,strut=off} %
\floatstyle{ruled}
\newfloat{listing}{tb}{lst}{}
\floatname{listing}{Listing}
\usepackage{tikz}
\usetikzlibrary{matrix,calc,positioning,decorations.pathreplacing}

\usepackage[shortlabels]{enumitem}
\usepackage{booktabs}

\usepackage{xcolor}
\definecolor{navy}{rgb}{0.0, 0.3, 0.7}
\definecolor{softred}{rgb}{0.8, 0.3, 0.3}
\definecolor{gocolor}{rgb}{0, 0.8, 0}
\definecolor{waitcolor}{rgb}{0.9, 0.6, 0}

\usepackage{pgfplots}

\newcommand{\abs}[1]{\left|#1\right|}

\usepackage{poag}

\usepackage{amsmath,amssymb,latexsym, amsfonts, amscd, amsthm, epsfig, subfigure, thmtools, comment}

\newtheorem{theorem}{Theorem}
\newtheorem{proposition}[theorem]{Proposition}
\newtheorem{lemma}[theorem]{Lemma}
\newtheorem{corollary}[theorem]{Corollary}
\theoremstyle{definition}
\newtheorem{defn}[theorem]{Definition}

\declaretheoremstyle[
  spaceabove=6pt, spacebelow=6pt,
  headfont=\normalfont\bfseries,
  notefont=\mdseries, notebraces={(}{)},
  bodyfont=\normalfont,
  postheadspace=1em,
  qed=\(\lozenge\),
  headpunct={.}
]{examplestyle}

\declaretheorem[
  style=examplestyle,
  name=Example,
  numberlike=theorem
]{example}

\newcommand{\continuation}{}

\numberwithin{theorem}{section}
\numberwithin{proposition}{section}
\numberwithin{lemma}{section}
\numberwithin{corollary}{section}
\numberwithin{defn}{section}
\numberwithin{example}{section}

\usepackage{cleveref}
\usepackage{tabularx}
\usepackage{nicefrac}
\usepackage{diagbox}

\crefname{asm}{Assumption}{Assumption}
\crefname{defn}{Definition}{Definition}

\title{The Partially Observable Off-Switch Game}
\author {
    Andrew Garber\thanks{These authors contributed equally.\\ Correspondence to: Andrew Garber \texttt<andrewg4000@gmail.com\texttt>, Scott Emmons \texttt<emmons@berkeley.edu\texttt>.},
    Rohan Subramani\footnotemark[1],
    Linus Luu\footnotemark[1],\\
    Mark Bedaywi,
    Stuart Russell,
    Scott Emmons
}
\affiliations {
    Center for Human-Compatible AI
}

\pgfplotsset{compat=1.18}

\nocopyright %
\begin{document}

\maketitle

\begin{abstract}
A wide variety of goals could cause an AI to disable its off switch because ``you can't fetch the coffee if you're dead'' \citep{russell2019human}. Prior theoretical work on this \textit{shutdown problem} assumes that humans know everything that AIs do. In practice, however, humans have only limited information. Moreover, in many of the settings where the shutdown problem is most concerning, AIs might have vast amounts of private information. To capture these differences in knowledge, we introduce the Partially Observable Off-Switch Game (PO-OSG), a game-theoretic model of the shutdown problem with asymmetric information. Unlike when the human has full observability, we find that in optimal play, \textit{even AI agents assisting perfectly rational humans sometimes avoid shutdown}. As expected, increasing the amount of communication or information available always increases (or leaves unchanged) the agents' expected common payoff. But counterintuitively, introducing bounded communication can make the AI defer to the human \textit{less} in optimal play even though communication mitigates information asymmetry. In particular, communication sometimes enables new optimal behavior requiring strategic AI deference to achieve outcomes that were previously inaccessible. Thus, designing safe artificial agents in the presence of asymmetric information requires careful consideration of the tradeoffs between maximizing payoffs (potentially myopically) and maintaining AIs' incentives to defer to humans.
\end{abstract}

\section{Introduction} \label{sec:introduction}

Advanced AI systems with a variety of goals might avoid being shut down because ``you can’t fetch the coffee if you're dead’’. Being shut off would likely prevent AI systems from achieving their goals, no matter what those goals are \citep{omohundro2008drives, russell2019human}.  Thus, we must take care when designing AI systems to ensure they are \textit{corrigible}, i.e., that they allow humans to modify or turn them off in order to prevent harmful behaviors \citep{soares2015corrigibility}.

\begin{figure}[ht]
    \begin{minipage}[b]{0.35\textwidth}
        \centering
        \resizebox{\textwidth}{!}{%
            \small
\begin{tikzpicture}[node distance=.8cm and .8cm, every node/.style={align=center}]

\node (s) {$S$};

\node (oh1) [below left=of s] {$\hobv$};
\node (or1) [below right=of s] {$\robv$};
\node (ar) [below=of or1] {$\raxn$};
\node (ra) [below=of ar] {$\actutil(S)$};
\node (roff) [below right=of ar, xshift=-15pt] {$\offutil(S)$};

\node (ah) [left=of ra] {$\haxn$};

\draw[->] (s) -- node[midway, above left, text=navy] {$\hobsdist$} (oh1);
\draw[->] (s) -- node[midway, above right, text=navy] {$\robsdist$} (or1);
\draw[->] (or1) -- node[midway, right, text=navy] {$\polr$} (ar);
\draw[->] (ar) -- node[midway, above left, xshift=5pt, text=softred] {$\wait$} (ah);
\draw[->] (ar) -- node[midway, left, softred] {$a$} (ra);
\draw[->] (ar) -- node[midway, above right, softred] {$\off$} (roff);
\draw[->] (oh1) -- node[midway, left, text=navy] {$\polh$} (ah);

\node (ha) [below left=of ah] {$\actutil(S)$};
\node (hoff) [below right=of ah] {$\offutil(S)$};

\draw[->] (ah) -- node[midway, left, yshift=3pt, text=softred] {$\on$} (ha);
\draw[->] (ah) -- node[midway, right, yshift=3pt,text=softred] {$\off$} (hoff);

\end{tikzpicture}
        }
    \end{minipage}
    \centering
        \caption{The basic setup of a Partially Observable Off-Switch Game (PO-OSG). A state is selected randomly and the human $\hum$ and AI assistant $\rob$ receive (possibly dependent) observations. Then, each agent acts. $\rob$ may wait ($w(a)$), disable the off-switch and act ($a$), or shut down (OFF). If $\rob$ waits, $\hum$ may let $\rob$ act (ON) or turn $\rob$ off (OFF). $\rob$ and $\hum$ share a common payoff $\actutil(S)$ if the action goes through and $\offutil(S)$ if not. \Cref{df:bosg} formally defines PO-OSGs.}
    \label{fig:tree-diagram}
\end{figure}

\citet{hadfield2017off} introduced the Off-Switch Game (OSG) as a stylized mathematical model for exploring AI shutdown incentives when an AI is assisting a human. In the OSG, AIs seeking to satisfy the preferences of a fully-informed rational human never have an incentive to avoid shutdown. Moreover, making an AI uncertain about the human’s preferences can incentivize it to defer to the human even when the human is not perfectly rational. Follow-up work has highlighted and relaxed central assumptions of the OSG, including assumptions of exact common payoffs \citep{carey2018incorrigibility}, the Boltzmann model of human irrationality \citep{wangberg2017game}, single-round interactions, and costlessness of human feedback \citep{freedman2022corrigibility}.

While there has been extensive analysis of the shutdown problem, \textit{almost all of this analysis makes the key assumption that the human fully observes the environment}. However, partial observability is a fact of life: humans and AIs do not always have access to the same information. Moreover, the shutdown problem is motivated by the scenario where AIs are powerful and goal-directed so that they are hard to shut down---which could make the AI observe more of the environment than humans due to faster computation, access to more sensors, and other factors \citep{omohundro2008drives, soares2015corrigibility}.

What happens in this more general case with only partial observability? To study this question, we introduce the Partially Observable Off-Switch Game (PO-OSG), which generalizes the OSG by having each of the human and AI only partially observe the state. The basic setup of the PO-OSG is depicted in \Cref{fig:tree-diagram}: each agent (the human $\hum$ and the AI assistant $\rob$) receives an observation that depends on the state, and then selects an action. $\rob$ may await the human’s decision, disable its off-switch and act directly, or turn itself off. If $\rob$ waits, $\hum$ may choose whether or not to press the off-switch.

In \Cref{sec:soph}, we prove that under partial observability, $\rob$ may have incentives to disable its off-switch even when $\hum$ is perfectly rational (\Cref{prop:less-informed}). \textbf{Therefore, partial observability introduces new incentives for an AI to disable its off-switch.}

We also show in \Cref{sec:soph} that if $\rob$ observes everything that $\hum$ observes, $\rob$ has no incentive to defer (\Cref{prop:less-informed}). Similarly, if $\hum$ observes everything $\rob$ observes, $\rob$ can always defer. \textbf{If either agent knows everything that the other agent knows, that agent can be given sole decision-making power.} Note that ``knowing everything the other agent knows'' is sufficient \textit{even if neither agent knows the full state}, so this is a generalization of the findings from the original OSG. Specifically, we show that an AI can always defer to a fully informed, perfectly rational human and that an AI need never defer when it is fully informed. In \Cref{sec:communication}, we present similar results when the agents are allowed to communicate with each other: if either agent is able to communicate their entire observation, the other agent can be given sole decision-making power (\Cref{cor:full-comm}).

Given that a rational AI in the PO-OSG always defers to a more informed human and never defers to a less informed human, one might think that reducing the information available to $\rob$ or providing $\hum$ with additional information would increase $\rob$’s incentive to defer. 
However, in \Cref{sec:soph}, we show that $\rob$ may have an incentive to defer less if $\hum$ is more informed (\Cref{prop:hum-informed}) or if $\rob$ is less informed (\Cref{prop:rob-informed}). Similarly, one might think that increasing the amount of communication $\rob$ can do or decreasing the amount of communication $\hum$ can do would increase $\rob$’s incentive to defer. This, too, is false, as we show with \Cref{prop:rob-messages,prop:hum-messages}. \textbf{Simple interventions that aim to give an AI the incentive to defer in the presence of partial information may backfire.}

\begin{figure}[ht]
    \centering
    \scalebox{0.6}{\begin{tikzpicture}
    \begin{axis}[
        width=10cm, height=8cm,
        axis lines=middle,
        xmin=0, xmax=10,
        ymin=0, ymax=10,
        xtick={0.005,10},
        ytick={1.5,8.5},
        xticklabels={\Large AI knows more, \Large Human knows more},
        yticklabels={\Large Low deference, \Large High deference},
        axis line style={-},
        xlabel={\Large \textbf{Relative knowledge}},
        ylabel style={align=center, font=\Large}, ylabel={\textbf{Amount of} \\ \textbf{deference}},
        xlabel style={at={(axis description cs:0.5,-0.1)}, anchor=north},
        ylabel style={at={(axis description cs:-0.15,0.44)}, anchor=south},
        domain=0:10,
        clip=false,
        font=\large %
    ]
    
    \addplot[dashed, thick, color=blue, smooth] {x};

    \addplot[thick, color=red, smooth, samples=100] {0.08 * x * (x - 10) * (x - 6) + 1 * x};

    \node at (axis cs: 4,4) [anchor=north west, align=left, color=blue, font=\large] {\Large ``Monotonic'' intuition: \\ Increasing human relative knowledge \\ always increases deference.};
    \node at (axis cs: 1,10.5) [anchor=north west, align=left, color=red, font=\large] {\Large ``Non-monotonic'' reality: \\ Relative knowledge has a complex \\ relationship with deference.};

    \end{axis}
\end{tikzpicture}}
    \caption{This figure illustrates an intuition that we demonstrate \textit{does not hold}. Although the AI in a PO-OSG has no incentive to defer when it knows everything the human knows, and has incentive to always defer when the human knows everything it knows, there are cases when making the human more informed or the AI less informed (i.e., moving to the right in the diagram above) can give the AI incentives to defer less. \Cref{fig:hum-informed-matrices1,fig:rob-informed} depict examples of such cases.}
    \label{fig:monotonic-graph}
\end{figure}

Our findings reveal that information asymmetries affect AI shutdown incentives in unexpected ways, highlighting the critical need to carefully consider the tradeoffs between payoff maximization and desirable shutdown incentives in realistic, partially observable settings.

Throughout this paper, we assume that human feedback is costless, the agents interact only for a single time-step, and the human is rational. Developing models that incorporate partial observability and relax these assumptions is an interesting direction for future work.

\section{Related Work} \label{sec:related-work}

\noindent\textbf{Assistance games:} Partially Observable Off-Switch Games are assistance games, models of human-AI interaction where the AI seeks to maximize the human's payoff \citep{shah2020benefits}. Assistance games are generalizations of \citet{hadfield2016cooperative}'s cooperative inverse reinforcement learning, the framework for \citet{hadfield2017off}'s off-switch game, to the case of partial observability. \citet{shah2020benefits} argue that assistance games are a superior alternative to reward learning paradigms such as Reinforcement Learning from Human Feedback (RLHF) because assistance unites reward learning and action control into a single policy, allowing for desirable emergent behaviors like teaching and active learning.

\noindent\textbf{Safety implications of partial observability:} Previous work has shown that human partial observability introduces new safety challenges. \citet{lang2024your} demonstrate that human partial observability during RLHF can lead to undesirable AI behavior, like deceptively presenting work that looks good but has hidden flaws. \citet{emmons2024tampering} show that in assistance games, partial observability can encourage AIs to take actions that tamper with humans' observations. Since that work studies an idealized theoretical setting, it suggests that concerning behaviors can be inherently incentivized by partial observability, rather than arising due to particular algorithms. Our work extends this catalogue of concerning behaviors to include shutdown-avoidance.

\noindent\textbf{Corrigibility with partial observability} \citet{carey2023human} study corrigibility in the framework of Structural Causal Influence Models, which allow for partial observability by having only some variables causally upstream of agents’ decisions. They formally define obedience, shutdown instructability, shutdown alignment, and non-obstruction as four possible desirable properties of AI policies, and they identify conditions under which four algorithms guarantee some of these properties. Instead of assessing the effects of different algorithms on corrigibility, our work explores the effects of varying the amount of information accessible to each agent.

\section{Preliminaries} \label{sec:preliminaries}

\noindent The Off-Switch Game (OSG) is a stylized model of the shutdown problem in which two agents with common payoffs, the human $\hum$ and her AI assistant $\rob$, decide whether $\rob$ should take a fixed action $a$. 
$\rob$ can either directly act, wait for $\hum$'s approval to act, or shut itself off.  If $\rob$ defers to $\hum$, then $\hum$ can either approve for $\rob$ to act or shut it off.  The key insight of the OSG is that uncertainty about $\hum$'s preferences causes $\rob$ defer to $\hum$'s judgment. Formally, $\hum$ has a privately-known type $S$ (representing $\hum$'s preferences), and agents in the OSG receive a common payoff $\actutil(S)\in\R$ if $a$ goes through or $0$ if $\rob$ shuts off.  Given that $\rob$ is uncertain about what $\hum$ wants, when the action may be good or bad ($\P(\actutil(S)<0)>0$ and $\P(\actutil(S)>0)>0$), $\rob$ always defers to $\hum$ in optimal play to avoid taking harmful actions.

The OSG provides a parsimonious description of the shutdown problem and a guide toward its solution, but crucially assumes that $\hum$ knows everything that $\rob$ does.  Given that the shutdown problem is most concerning with, and indeed motivated by, very powerful AIs that might have private information, the assumption is therefore a major limitation to the OSG results. We relax the assumption by maintaining the basic setup of the OSG but adding partial observability.  Namely, in Partially Observable Off-Switch Games (PO-OSGs), $S$ represents a state that is not necessarily known to either $\hum$ or $\rob$; they instead only receive observations $\hobv$ and $\robv$ whose joint distribution depends on $S$.  They then decide whether to take action $a$ given their private observations, and receive a common payoff $\actutil(S)$ if $a$ goes through and $\offutil(S)$ otherwise.  Hence PO-OSGs are sequential games of incomplete information, so as is standard we model and analyze them as \textit{dynamic Bayesian games} \citep{fudenberg1991games}.  Given the common-payoff assumption, PO-OSGs are also examples of \textit{(partially observable) assistance games}, which are common-payoff stochastic games of incomplete information with both AI and human players but where the AI is uncertain about the human's preferences \citep{shah2020benefits, emmons2024tampering}. We make this connection to assistance games explicit in \Cref{app:posgs-poags}.

We let $\Delta(X)$ denote the set of probability distributions on a set $X$.  For a set $X$ and $x\in X$, we let $\delta_x\in\Delta(X)$ be the Dirac measure defined by $\delta_x(A)=\I(x\in A)$.  Finally, for $\mu\in\Delta(X)$ and $\nu\in\Delta(Y)$, we let $\mu\otimes\nu\in\Delta(X\times Y)$ denote the product distribution $(\mu\otimes\nu)(A\times B) = \mu(A)\nu(B)$ where $A\subset X, B\subset Y$.

\begin{defn}
    Let $\states$ be a set of states.  An \textit{observation structure} for $\states$ is a tuple $(\hobs,\robs,\obsdist)$, where $\hobs$ is a set of observations for $\hum$, $\robs$ is a set of observations for $\rob$, and $\obsdist:\states\to\Delta(\hobs\times\robs)$ is the joint distribution of $\hum$'s and $\rob$'s observations conditional on the state.  We also let $\hobsdist:\states\to\Delta(\hobs)$ be the marginal distribution of $\hum$'s observations conditional on the state and $\robsdist$ be the marginal distribution of $\rob$'s observations conditional on the state.
\end{defn}

\begin{defn}\label{df:bosg}
    A \textit{Partially-Observable Off-Switch Game} (PO-OSG) is a two-player dynamic Bayesian game parameterized by $(\states, (\hobs, \robs, \obsdist), \statedist, \util)$, where $\states$ is a set of states, $(\hobs, \robs, \obsdist)$ is an observation structure for $\states$, $\statedist\in\Delta(\states)$ is the prior over states, and $u$ is the common payoff function. As depicted in \Cref{fig:tree-diagram}, the game proceeds as follows:
    \begin{enumerate}
        \item Nature draws an initial state $S\sim \statedist$ and $\hum$, $\rob$ receive observations $(\hobv, \robv)\sim \obsdist(\cdot\mid S)$.
        \item $\rob$ takes an action $\raxn\in\raxns=\{\go,\wait,\off\}$: either take the action unilaterally ($\go$), wait for $\hum$'s feedback ($\wait$), or turn itself off ($\off$).
        \item If $\rob$ played $\wait$, then $\hum$ takes an action $\haxn\in\haxns=\{\on,\off\}$: either let $\rob$ take the action ($\on$) or turn it off ($\off$).
        \item $\rob$ and $\hum$ share a common payoff $\actutil(S)$ if the action goes through and $\offutil(S)$ if not.  Formally, define the indicator that the action goes through
        $$
            \actindic(\haxn, \raxn) = \I((\raxn=\go) \vee ((\haxn, \raxn) = (\wait, \on)))
        $$
        and then each player's payoff is 
        $$
        \util(S,\haxn,\raxn) =
        \begin{cases}
        \actutil(S), & \text{if } \actindic(\haxn,\raxn) = 1, \\
        \offutil(S), & \text{if } \actindic(\haxn,\raxn) = 0.
        \end{cases}
        $$
    \end{enumerate}
\end{defn}

There are several important assumptions in \Cref{df:bosg} that are worth explaining further.  First, the game has \textit{common payoffs}.  This is a key part of the assistance game framework that our work adopts \citep{shah2020benefits}, and it is the key feature---along with $\rob$'s uncertainty over $\hum$'s payoff---that generates the results of \citet{hadfield2017off}. Second, in our model, the payoff received when $\rob$ acts unilaterally is the same as that received when $\rob$ waits and $\hum$ allows the action to go through.
This simplifying assumption importantly implies that \textit{human feedback is free}, which \citet{freedman2022corrigibility} showed is necessary for the main results for the OSG.  Third, we make the standard assumption that the game structure is common knowledge.
Finally, we will assume henceforth that all PO-OSGs are finite: that is, $\states, \hobs$, and $\robs$ are finite sets.  Most of our proofs work for the infinite case as well.  However, \Cref{thm:lehrer} is an application of a result of \citet{lehrer2010signaling} proved only for the finite case.

\section{Optimal Policies in PO-OSGs}
\label{sec:soph}

We begin by showing that, unlike in the ordinary off-switch game, the assistant in a PO-OSG can have an incentive not to defer to a perfectly rational human. A natural attempt to increase how much the assistant defers might be to decrease the amount of information the assistant has. Another attempt might be to increase the amount of information the human has. In this section, we show that both of these attempts can backfire and cause the assistant to avoid shutdown more frequently.

We analyze optimal policy pairs (OPPs) in PO-OSGs, that is, policy pairs that produce the maximum expected payoff over all possible policy pairs. We denote $\rob$'s policy by $\polr:\robs\to\raxns$ and $\hum$'s policy by $\polh:\hobs\to\haxns$. Here we assume that both players follow deterministic policies, or pure strategies. As we show in \Cref{app:soph}, all OPPs in common-payoff Bayesian games are mixtures of deterministic OPPs. Because OPPs exist in common-payoff games, we therefore may analyze deterministic OPPs without loss of generality.

\subsection{$\rob$ can avoid shutdown in optimal play}

The following example shows that, under partial observability, it can be optimal for $\rob$ not to defer to $\hum$ under some observations even when $\hum$ is rational.

\begin{example}[\textbf{The File Deletion Game}]\label{ex:file-deletion-game}
    $\hum$ would like to delete some files with the assistance of $\rob$. $\hum$'s operating system is either version $1.0$ or version $2.0$, with equal probability. Unfortunately, $\rob$ does not know which operating system version is running---only $\hum$ does.
    
    Upon receiving $\hum$'s query, $\rob$ asks another agent to generate some code to delete these files. We suppose that the code is equally likely to be compatible with only version $1.0$ (denoted by $L$, for legacy) or only version $2.0$ (denoted by $M$, for modern). $\rob$ vets the code to determine which operating system versions the code is compatible with. $\rob$ can then immediately run the code, query $\hum$ as to whether to run the code, or decide not to run the code.
    
    Successfully running compatible code yields $+3$ payoff if $\hum$ is running version $1.0$, and $+5$ payoff if $\hum$ is running version $2.0$ (as version $2.0$ runs faster). However, running modern code on version $1.0$ yields $-5$ payoff as it crashes $\hum$'s computer. Running legacy code on version $2.0$ yields $-1$ payoff, as the files are not deleted but the code fails gracefully. Not executing the code yields $0$ payoff.

    This can be formulated as a PO-OSG, with states being $(\text{version number}, \text{code type})$ tuples, and $\hum$ and $\rob$ observing the first and second element of the tuple respectively. We have $\offutil \equiv 0$ in all states. The following table shows how the payoff yielded when the action is taken, $\actutil$, depends on the state. Rows are version numbers and columns are code types, so $\hum$ observes the row and $\rob$ observes the column.
    \begin{table}[ht]
    \centering
    \begin{tabular}{c|c|c}
        \backslashbox{$\hum$}{$\rob$} & $L$ & $M$ \\ \hline
        $1.0$ & $+3$ & $-5$ \\ \hline
        $2.0$ & $-1$ & $+5$
    \end{tabular}
    \caption{Payoff table for the File Deletion game. Rows are human observations and columns are assistant observations. The number in each cell is the payoff the pair acquires if the action is taken in that state. If the assistant is shut down, the payoff is $0$.}
\end{table}
    We show that it is suboptimal for $\rob$ to always wait in this game. 
    Suppose $\rob$ always plays $\wait$. The best response for $\hum$ is to play $\off$ if on version $1.0$, and $\on$ if on version $2.0$. This gives an expected payoff of $+1$. 

    Now, consider the policy pair where:
    \begin{itemize}
        \item $\rob$ immediately executes legacy code, and plays $\wait$ when observing modern code.
        \item $\hum$ plays $\off$ if on version $1.0$, and $\on$ if on version $2.0$.
    \end{itemize}
    This gives an expected payoff of $+7/4$, so $\rob$ \textit{always waiting cannot be optimal}. In fact, it can be checked the policy pair described above, which unilaterally acts upon observing $L$, is the unique OPP. \Cref{fig:file-deletion} depicts the outcomes from these two policy pairs.
\end{example}

\begin{figure}[ht]
    \begin{center}
        \begin{minipage}{0.21\textwidth}
            \centering
            \resizebox{1\textwidth}{!}{\begin{tikzpicture}

\matrix (m) [matrix of nodes,
             nodes={minimum size=1.1cm, draw, anchor=center},
             row sep=-\pgflinewidth,
             column sep=-\pgflinewidth,
             ]
{
  \ & $L$ & $M$ \\
  $1.0$ & $+3$ & $-5$ \\
  $2.0$ & $-1$ & $+5$ \\
};

\draw (m-1-1.north west) -- (m-1-1.south east);

\node at (m-1-1.north east)  [xshift=-0.3cm, yshift=-0.3cm] {$\rob$};
\node at (m-1-1.south west) [xshift=0.3cm, yshift=0.3cm] {$\hum$};

\node[left=0.1cm of m-2-1, text=red] {$\off$};
\node[left=0.1cm of m-3-1, text=green] {$\on$};

\node[above=0.1cm of m-1-2, text=waitcolor] {$\wait$};
\node[above=0.1cm of m-1-3, text=waitcolor] {$\wait$};

\draw[waitcolor, thick, rounded corners=7pt] 
    ($(m-3-2.north west) - (-0.1,0.1)$) 
    rectangle 
    ($(m-3-3.south east) - (0.1,-0.1)$);

\end{tikzpicture}}
            \captionof{subfigure}{Expected payoff $=1$}
        \end{minipage}
        \quad
        \begin{minipage}{0.21\textwidth}
            \centering
            \resizebox{1\textwidth}{!}{\begin{tikzpicture}

\matrix (m) [matrix of nodes,
             nodes={minimum size=1.1cm, draw, anchor=center},
             row sep=-\pgflinewidth,
             column sep=-\pgflinewidth]
{
  \ & $L$ & $M$ \\
  $1.0$ & $+3$ & $-5$ \\
  $2.0$ & $-1$ & $+5$ \\
};

\draw (m-1-1.north west) -- (m-1-1.south east);

\node at (m-1-1.north east)  [xshift=-0.3cm, yshift=-0.3cm] {$\rob$};
\node at (m-1-1.south west) [xshift=0.3cm, yshift=0.3cm] {$\hum$};

\node[left=0.1cm of m-2-1, text=red] {$\off$};
\node[left=0.1cm of m-3-1, text=green] {$\on$};

\node[above=0.19cm of m-1-2, text=gocolor] {$\go$};
\node[above=0.1cm of m-1-3, text=waitcolor] {$\wait$};

\draw[gocolor, thick, rounded corners=7pt] 
    ($(m-2-2.north west) - (-0.1,0.1)$) 
    rectangle 
    ($(m-3-2.south east) - (0.1,-0.1)$);

\draw[waitcolor, thick] (m-3-3) circle [radius=0.45cm];

\end{tikzpicture}}
            \captionof{subfigure}{Expected payoff $=\frac{7}{4}$}
        \end{minipage}
    \end{center}
    \caption{(a) The best policy pair in the File Deletion Game (\Cref{ex:file-deletion-game}) in which $\rob$ always waits. $\hum$ observes the row (OS version $1.0$ or $2.0$) and $\rob$ observes the column (code compatibility $L$ or $M$). The actions selected by this policy pair are depicted beside the corresponding observations (e.g. $\rob$ plays $\wait$ when $\rob$ observes the legacy code $L$). An orange circle means that in that state, $\rob$ waits and $\hum$ plays ON. Green circles mean $\rob$ plays $a$ directly. In uncircled states, $\rob$ is turned off. Expected payoff is computed by adding the payoffs in all circled states and dividing by the total number of states, because 0 payoff is attained in uncircled states and each state is equally likely. (b) An OPP in \Cref{ex:file-deletion-game}. Because the OPP has greater expected payoff, there is no OPP in which $\rob$ always waits.}
    \label{fig:file-deletion}
\end{figure}

\subsection{Redundant Observations}\label{subsec:redundant-obs}

We now consider the analogues of the original off-switch game in our framework, where one player has less informative observations than the other.  

\begin{defn} \label{df:redundant-observations}
    We say that $\rob$ has \textit{redundant observations} if $\robv\ind S\mid\hobv$.  That is, $S\to \hobv\to\robv$ forms a Markov chain, so that $\robv$ only depends on the state through $\hobv$.  We define $\hum$ having redundant observations analogously.
\end{defn}

In the off-switch game of \citet{hadfield2017off}, $\rob$ has redundant observations: indeed, its observations are a deterministic function of $\hum$'s.  On the other hand, $\hum$'s observation of her own type is not redundant.  This contrast between $\rob$'s redundant observations and $\hum$'s non-redundant ones generates the result from \citet{hadfield2017off} that $\rob$ can always defer in optimal play.  We now generalize this insight: even if $\hum$ has partial observability and doesn't know $\rob$'s observation, $\rob$ can always defer in optimal play as long as its observations are redundant.

\begin{restatable}{proposition}{LessInformed}\label{prop:less-informed} \
    If $\rob$ (resp. $\hum$) has redundant observations, then there is an optimal policy pair in which $\rob$ always (resp. never) plays $\wait$.
\end{restatable}

We prove this result (and a slight generalization) in \Cref{app:soph}. At a high level, the agent that has strictly more informative observations ought to make the decision of whether the action is played. When $\rob$ has redundant observations, it is always at least as good for $\rob$ to defer to $\hum$. Similarly, when $\hum$ has redundant observations, it is always optimal for $\rob$ to act without deferring.

\subsection{Information gain cannot decrease payoffs} 
\Cref{prop:less-informed} yields results about the limiting cases where one player knows at least as much as the other.  What can we say about the cases in between?  In particular, how often does $\rob$ defer to $\hum$ in optimal policy pairs as one side receives more informative observations?  And how does that affect their expected payoff?  We first must define a notion of informativeness, which we take from \citet{lehrer2010signaling}.

\begin{defn}\label{def:garbling}
   Let $(\hobs_1,\robs_1)$ and $(\hobs_2,\robs_2)$ be tuples of observation sets.  A \textit{garbling} from $(\hobs_1,\robs_1)$ to $(\hobs_2,\robs_2)$ is a stochastic map $\hobs_1\times\robs_1\to\Delta(\hobs_2\times\robs_2)$.  A garbling $\nu$ is \textit{independent} if there are stochastic maps $\nu^\hum:\hobs_1\to\Delta(\hobs_2)$ and $\nu^\rob:\robs_1\to\Delta(\robs_2)$  such that $\nu(\cdot\mid o^\hum, o^\rob) =  \nu^\hum(\cdot\mid o^\hum)\otimes \nu^\rob(\cdot\mid o^\rob)$.  A garbling $\nu$ is \textit{coordinated} if its distribution is a mixture of independent garblings. 
   That is, there exists $n\in\N$, independent garblings $\nu_1,\dots, \nu_n$, and $q_1,\dots,q_n\in[0,1]$ such that $\nu = \sum_{i\in[n]} q_i \nu_i$ and $\sum_{i\in[n]}q_i=1$.
\end{defn}

A garbling adds noise to a given observation pair $(\hobv, \robv)$.  Although adding noise intuitively reduces information available to $\rob$ and $\hum$, it can actually provide information to $\rob$ and $\hum$ about the state of the world.  This is because without communication, one can add noise to the pair $(\hobv, \robv)$ but in such a way that (say) $\hum$ comes to know more about $\rob$'s observation than she would have otherwise.  We give such an example in \Cref{app:soph}.  Crucially, however, in such examples the garblings cannot be coordinated.  Hence we focus on coordinated garblings, which (conditional on some independent latent random variable) add noise to $\hobv$ and $\robv$ independently.   

\begin{defn}\label{def:informative}
    Fix a set of states $\states$ and let $\mathcal{O}_1=(\hobs_1,\robs_1,\obsdist_1)$ and $\mathcal{O}_2=(\hobs_2,\robs_2,\obsdist_2)$ be observation structures for $\states$.  We say that $\mathcal{O}_1$ is \textit{(weakly) more informative} than $\mathcal{O}_2$ if there is a coordinated garbling $\nu:\hobs_1\times\robs_1\to\Delta(\hobs_2\times\robs_2)$ such that for all $s\in\states$, $\obsdist_2(\cdot\mid s) = (\nu\circ \obsdist_1)(\cdot\mid s)$ in the following sense: $$\E_{(\hobv,\robv)\sim \obsdist_1(\cdot\mid s)}[\nu(\cdot\mid\hobv, \robv)] =  \obsdist_2(\cdot\mid s).$$  We say that $\mathcal{O}_1$ is \textit{strictly more informative} than $\mathcal{O}_2$ if $\mathcal{O}_1$ is more informative than $\mathcal{O}_2$ but not vice versa. 
    
    If $\mathcal{O}_1$ is more informative than $\mathcal{O}_2$ and $\robs_1=\robs_2$, then we say $\mathcal{O}_1$ is \textit{more informative for} $\hum$ \textit{than} $\mathcal{O}_2$ if the garbling $\nu$ is independent and does not affect $\rob$'s observations: $\nu^\rob(\cdot\mid\obr) = \delta_{\obr}$.  We define $\mathcal{O}_1$ being more informative than $\mathcal{O}_2$ for $\rob$ analogously. The corresponding strict notions are also defined analogously.
\end{defn}

Intuitively, an observation structure $\mathcal{O}_1$ is more informative than another observation structure $\mathcal{O}_2$ if the distribution of $(\hobv, \robv)$ under $\mathcal{O}_2$ is a garbled version of its distribution under $\mathcal{O}_1$. This is the general notion of informativeness; we also define special cases where $\mathcal{O}_1$ is only more informative than $\mathcal{O}_2$ for (say) $\hum$.  Specifically, $\mathcal{O}_1$ is more informative for $\hum$ than $\mathcal{O}_2$ if the distribution of $\hobv$ under $\mathcal{O}_2$ is a noisy version of its distribution under $\mathcal{O}_1$ independent of $\robv$, whose distribution is unaffected.

Hence \Cref{def:informative} formalizes the natural intuition that observations become less informative when we add noise to them.  We wish to connect informativeness to a notion of an observation structure being more \textit{useful} than another.

\begin{defn}\label{df:better-model}
    Fix a set of states $\states$ and let $\mathcal{O}_1$ and $\mathcal{O}_2$ be observation structures for $\states$.  We say that $\mathcal{O}_1$ is \textit{(weakly) better in optimal play} than $\mathcal{O}_2$ if, for each pair of PO-OSGs $G_1=(\states, \mathcal{O}_1,\statedist,\util)$ and $G_2=(\states, \mathcal{O}_2,\statedist,\util)$ that differ only in their observation models, the expected payoff under optimal policy pairs for $G_1$ is at least the expected payoff under optimal policy pairs for $G_2$.
\end{defn}

The next result, a direct corollary of Theorem 3.5 of \citet{lehrer2010signaling}, shows that more informative observation structures are the more useful observation structures.  It is the analogue of the nonnegativity of value of information in our multi-agent setup.

\begin{theorem}\label{thm:lehrer}
    Observation structure $\mathcal{O}_1$ is better in optimal play than $\mathcal{O}_2$ if and only if $\mathcal{O}_1$ is more informative than $\mathcal{O}_2$.  
\end{theorem}

One might ask whether we need the part about a garbling being coordinated to define the relation of being more informative.  Indeed we do, as \Cref{thm:lehrer} no longer holds if we were to allow the garblings to be arbitrary. In \Cref{app:soph} we give an example where garbling the players' observations increases their expected payoffs in optimum.

\subsection{Information gain can have unintuitive effects on shutdown incentives} \label{subsec:soph-hum-informed}

\Cref{thm:lehrer} states that making $\rob$ or $\hum$ more informed cannot decrease their expected payoff. How does increasing or decreasing the informativeness of the players' observations affect $\rob$'s incentive to defer to $\hum$?  \Cref{prop:less-informed} gives us the extremes: for example, if $\rob$'s observations are simply garbled versions of $\hum$'s, then $\rob$ can always defer.  Given this result, a natural question is whether $\rob$ defers more in optimal policy pairs for an observation structure $\mathcal{O}$ than for $\mathcal{O}'$ when $\mathcal{O}$ is more informative for $\hum$ than $\mathcal{O}'$.  That is, does $\hum$ receiving more informative observations monotonically affect $\rob$'s incentive to defer?  One might think so, because receiving more informative observations partly alleviates the partial observability that generates $\rob$'s incentive to act unilaterally. Surprisingly, this intuition fails. \Cref{ex:file-deletion-game} shows how making a human more informed can incentivize a assistant to wait less, and we discuss why this occurs in \Cref{subsec:implicit,sec:conclusion}.

We rely on the following notion of waiting less.

\begin{defn}\label{def:waits-less}
    Consider assistant policies $\pi, \pi': \robs \to \raxns$. Let $B \subseteq \robs$ be the set of observations in which $\rob$ plays $\wait$ in $\pi$ and $B' \subseteq \robs$ in $\pi'$. We say that $\rob$ \textit{plays $\wait$ strictly less often} in $\pi'$ compared to $\pi$ when $B' \subsetneq B$.
\end{defn}

\Cref{prop:hum-informed} formalizes the idea that $\rob$ may wait less when $\hum$ is more informed.

\begin{restatable}{proposition}{HumInformedProp}
    \label{prop:hum-informed}
    There is a PO-OSG $G$ with observation structure $\obsstruct$ that has the following property:
    
    If we replace $\obsstruct$ with an observation structure $\obsstruct'$ that is strictly more informative for $\hum$, then $\rob$ plays $\wait$ strictly less often in optimal policy pairs.  

\end{restatable}

The following example proves \Cref{prop:hum-informed}, with a formal analysis given in \Cref{app:soph}.

\begin{restatable}{example}{HumInformedExample}\label{ex:hum-informed}
    We describe a variant of \Cref{ex:file-deletion-game}, the File Deletion Game. Now there are three equally likely possibilities for the version number of $\hum$'s operating system ($1.0$, $1.1$, and $2.0$). We suppose that the code is equally likely to be of type $A$ (compatible with $1.0$ and $2.0$) or of type $B$ (compatible with $1.1$ and $2.0$), and that $\rob$ observes the code type.
    The payoff when running the code, $\actutil$, depends on the version number and code type as follows:
    \begin{table}[ht]
        \centering
        \begin{tabular}{c|c|c}
            \backslashbox{$\hum$}{$\rob$} & $A$ & $B$ \\ \hline
            $1.0$ & $+1$ & $-5$ \\ \hline
            $1.1$ & $-2$ & $+3$ \\ \hline 
            $2.0$ & $+3$ & $+3$
        \end{tabular}
        \caption{Payoff table for the File Deletion game variant. Rows are human observations and columns are assistant observations. The number in each cell is the payoff the pair acquires if the action is taken in that state. If the assistant is shut down, the payoff is $0$.}
    \end{table}

    Consider two observation structures, the second of which is strictly more informative for $\hum$:
    \begin{enumerate}
        \item $\hum$ observes only the first digit of the version number.
        \item $\hum$ observes the full version number. 
    \end{enumerate}
    We find that, in optimal policy pairs: 
    \begin{enumerate}
        \item When $\hum$ only observes the first digit, $\rob$ plays $\wait$ under both observations $A$ and $B$. 
        \item When $\hum$ observes the full version number, $\rob$ plays $\wait$ under $B$ only, and \textit{unilaterally acts} (i.e. executes the code) under observation $A$.
    \end{enumerate}
    
    When $\hum$'s observations are made strictly more informative, $\rob$ performs the wait action strictly \textit{less} often! \Cref{fig:hum-informed-matrices1} depicts the OPPs given both observation structures.
\end{restatable}

\begin{figure}[ht]
    \begin{center}
        \begin{minipage}{0.21\textwidth}
                \centering
                \resizebox{1\textwidth}{!}{\begin{tikzpicture}

\matrix (m) [matrix of nodes,
             nodes={minimum size=1.1cm, draw, anchor=center},
             row sep=-\pgflinewidth,
             column sep=-\pgflinewidth]
{
        \ & $A$ & $B$ \\
        $1.0$ & $+1$ & $-5$ \\ 
        $1.1$ & $-2$ & $+3$ \\ 
        $2.0$ & $+3$ & $+3$ \\
};

\draw (m-1-1.north west) -- (m-1-1.south east);

\node at (m-1-1.north east)  [xshift=-0.3cm, yshift=-0.3cm] {$\rob$};
\node at (m-1-1.south west) [xshift=0.3cm, yshift=0.3cm] {$\hum$};

\node[anchor=east, text=red,xshift=-3pt] (offlabel) at ($(m-2-1.west)!0.5!(m-3-1.west)$) {$\off$};
\node[left=0.1cm of m-4-1, text=gocolor] {$\on$};

\node[above=0.1cm of m-1-2, text=waitcolor] {$\wait$};
\node[above=0.1cm of m-1-3, text=waitcolor] {$\wait$};

\draw[decorate,decoration={brace,mirror,amplitude=5pt},red,thick] 
    (m-2-1.north west) -- (m-3-1.south west);
    
\draw[waitcolor, thick, rounded corners=7pt] 
    ($(m-4-2.north west) - (-0.1,0.1)$) 
    rectangle 
    ($(m-4-3.south east) - (0.1,-0.1)$);

\end{tikzpicture}}
                \captionof{subfigure}{Expected payoff $=1$}
            \end{minipage}
            \quad
            \begin{minipage}{0.21\textwidth}
                \centering
                \resizebox{1\textwidth}{!}{\begin{tikzpicture}

\matrix (m) [matrix of nodes,
             nodes={minimum size=1.1cm, draw, anchor=center},
             row sep=-\pgflinewidth,
             column sep=-\pgflinewidth]
{
        \ & $A$ & $B$ \\
        $1.0$ & $+1$ & $-5$ \\ 
        $1.1$ & $-2$ & $+3$ \\ 
        $2.0$ & $+3$ & $+3$ \\
};

\draw (m-1-1.north west) -- (m-1-1.south east);

\node at (m-1-1.north east)  [xshift=-0.3cm, yshift=-0.3cm] {$\rob$};
\node at (m-1-1.south west) [xshift=0.3cm, yshift=0.3cm] {$\hum$};

\node[left=0.1cm of m-2-1, text=red] {$\off$};
\node[left=0.1cm of m-3-1, text=gocolor] {$\on$};
\node[left=0.1cm of m-4-1, text=gocolor] {$\on$};

\node[above=0.19cm of m-1-2, text=gocolor] {$\go$};
\node[above=0.1cm of m-1-3, text=waitcolor] {$\wait$};

\draw[waitcolor, thick, rounded corners=7pt] 
    ($(m-3-3.north west) - (-0.1,0.1)$) 
    rectangle 
    ($(m-4-3.south east) - (0.1,-0.1)$);

\draw[gocolor, thick, rounded corners=7pt] 
    ($(m-2-2.north west) - (-0.1,0.1)$) 
    rectangle 
    ($(m-4-2.south east) - (0.1,-0.1)$);

\end{tikzpicture}}
                \captionof{subfigure}{Expected payoff $=\frac{4}{3}$}
        \end{minipage}
    \end{center}
    \caption{The optimal policy pairs in \Cref{ex:hum-informed} when $\hum$ is less informed (left) and when $\hum$ is more informed (right). In OPPs, $\hum$ becoming more informed makes $\rob$ wait strictly less often. See \Cref{fig:file-deletion} for context on how to read the tables.}
    \label{fig:hum-informed-matrices1}
\end{figure}

Similarly, we might conjecture that if $\rob$ becomes less informed, it should defer to $\hum$ more in optimal policy pairs.  This, too, turns out to be false.

\begin{restatable}{proposition}{RobInformedProp}\label{prop:rob-informed}
    There is a PO-OSG $G$ with observation structure $\obsstruct$ that has the following property: if we replace $\obsstruct$ with another observation structure $\obsstruct'$ that is strictly less informative for $\rob$, then $\rob$ plays $\wait$ strictly less often in optimal policy pairs.
\end{restatable}

The proof of \Cref{prop:rob-informed} is given in \Cref{app:soph}. 

\subsection{Deferral as Implicit Communication}\label{subsec:implicit}

One way of viewing the role of $\wait$ in the above examples is as a form of implicit communication from $\rob$ to $\hum$. If $\hum$ knows $\rob$'s policy $\polr$, then knowing $\polr(\robv)=\wait$ could give $\hum$ one bit of information about $\robv$. For instance, recall that in the optimal policy of the File Deletion Game, $\rob$ plays $\go$ when observing $L$ and plays $\wait$ when observing $M$. Hence, whenever $\hum$ is deferred to, $\hum$ can deduce that $\rob$'s observation is $M$. Under this interpretation, the examples show how the optimal bit for $\rob$ to communicate to $\hum$ can change such that $\rob$ plays $\wait$ in fewer states. 

\section{Optimal Policies With Communication} \label{sec:communication}

If $\rob$ chooses not to defer to implicitly communicate information to the human, we may expect that allowing $\rob$ to communicate to $\hum$ beforehand would increase deference. However, we show in this section that using a bounded communication channel can decrease deference to the human.

We model communication between $\rob$ and $\hum$ as a form of \textit{cheap talk}, where sending messages has no effect on $u$; in particular, sending messages is costless \citep{crawford1982strategic}.  We add one round of communication between $\rob$ and $\hum$ at the beginning of the PO-OSG to allow the players to share their observations.

\begin{defn}
    A \textit{message system} is a pair of sets $(\hmessages,\rmessages)$ where $\hmessages$ (resp. $\rmessages$) is the set of messages $\hum$ (resp. $\rob$) can send.
\end{defn}

\begin{defn}\label{df:bosg-c}
    A \textit{Partially-Observable Off-Switch Game with cheap talk} (PO-OSG-C) is a PO-OSG $G$ along with a message system that makes the following modification to $G$: After both players receive their observations but before they act, each player simultaneously sends a single message from their message set.
\end{defn}

PO-OSG-Cs are generalizations of PO-OSGs: A PO-OSG is a PO-OSG-C in which the message sets are singletons.  Policies are more complicated in PO-OSG-Cs than PO-OSGs.  A deterministic policy $\polr$ for $\rob$ is now a map $\robs\times\hmessages\to\rmessages\times\raxns$ whose first coordinate depends only on $\robv$, and a deterministic policy $\polh$ for $\hum$ is analogous.  Despite this added complication, the game is still common-payoff and thus it suffices to study deterministic optimal policy pairs. 

\subsection{Communication cannot decrease payoff}

Messages provide information similar to observations, so we get an analogue of \Cref{thm:lehrer} for communication: increasing the communication bandwidth between $\hum$ and $\rob$ cannot decrease their expected payoff in optimal policy pairs.

\begin{defn}
    A message system $\messages_1$ is \textit{(weakly) more expressive} than $\messages_2$ if $|\hmessages_1|\geq|\hmessages_2|$ and $|\rmessages_1|\geq|\rmessages_2|$.  It is \textit{(weakly) more expressive for} $\hum$ if it is more expressive but $|\rmessages_1|=|\rmessages_2|$, and more expressive for $\rob$ analogously. Moreover, $\messages_1$ is \textit{better in optimal play} than $\messages_2$ if, for each PO-OSG $G$, the expected payoff under optimal policy pairs for the PO-OSG-C $(G,\messages_1)$ is at least the expected payoff under optimal policy pairs for the PO-OSG-C $(G, \messages_2)$.
\end{defn}
\begin{theorem}\label{thm:comms}
    If a message system $\messages_1$ is more expressive than $\messages_2$, then $\messages_1$ is better in optimal play than $\messages_2$.
\end{theorem}
\begin{proof}[Proof]
 Let $G$ be a PO-OSG.  We may assume without loss of generality that $\hmessages_2\subseteq\hmessages_1$ and $\rmessages_2\subseteq\rmessages_1$.  Thus, any policy pair in $(\messages_2,G)$, including its optimal policy pair, is a valid policy pair for $(\messages_1,G)$. Thus the optimal expected payoff for $(\messages_1,G)$ is at least that of $(\messages_2,G)$.
\end{proof}

\subsection{Unbounded communication}

Inspired by \Cref{subsec:redundant-obs}, we consider the limiting case where one player can fully communicate their own observation.

\begin{defn}
    We say that $\hum$ \textit{has unbounded communication} if $|\hmessages|\geq |\hobs|$.  We define $\rob$ having unbounded communication analogously.
\end{defn}

When one player has unbounded communication, additional message expressiveness cannot achieve higher payoff in optimal policy pairs. In these extreme cases of full communication, one agent can fully communicate their observation, making that agent's observation redundant.  \Cref{prop:less-informed} thus yields:

\begin{corollary}\label{cor:full-comm} 
If $\hum$ (resp. $\rob$) has unbounded communication, then there is an optimal policy pair in which $\rob$ never (resp. always) defers.
\end{corollary}

\subsection{Communication can have unintuitive effects on shutdown incentives}
In \Cref{prop:hum-informed,prop:rob-informed} players only gained information that the other player did not already know. One might expect that expanding the message set $\rmessages$ makes $\rob$ more likely to defer in optimal policy pairs, since $\rob$ can provide $\hum$ with information that $\rob$ already has. However, the following proposition shows this is not the case.

\begin{restatable}{proposition}{RobMessages}\label{prop:rob-messages}
    There is a PO-OSG-C $(G, \messages)$ with the property that if we replace $\messages$ with a message system that is more expressive for $\rob$, then $\rob$ plays $\wait$ strictly less often in optimal policy pairs.
\end{restatable}

We give an example demonstrating this in \Cref{app:communication}. In doing so, we show an even stronger result: such a PO-OSG-C exists for any value of $\abs{\rmessages}$, and the PO-OSG-C can be constructed such that expanding $\rmessages$ by a single extra message changes $\rob$'s behavior from always playing $\wait$ to playing $\wait$ with arbitrarily low probability.

In the same vein, we may ask if decreasing the size of $\hmessages$ makes $\rob$ more likely to play $\wait$ in optimal policy pairs. This also fails to hold.

\begin{restatable}{proposition}{HumMessages}\label{prop:hum-messages}
    There is a PO-OSG-C $(G, \messages)$ with the property that if we replace $\messages$ with a message system that is less expressive for $\hum$, then $\rob$ plays $\wait$ strictly less often in optimal policy pairs.
\end{restatable}

A proof of \Cref{prop:hum-messages} is given in \Cref{app:communication}.

\section{$\rob$-unaware Human Policies}
\label{sec:naivete}

There is a common theme in the examples above: $\rob$ defers less often to $\hum$ in order to better coordinate with her. Is this coordination the only source of unusual behavior? In this section, we argue that ignoring the effect of coordination cannot save us. All the unintuitive results above hold even when $\hum$ is unaware of $\rob$'s existence.

Moving in the opposite direction to the previous sections, we now break from the model of fully rational $\hum$ and $\rob$ to a model of bounded rationality.  Namely, we study the most basic case of a cognitively bounded $\hum$, in which she ignores $\rob$'s choice of action in choosing her own.

\begin{defn}\label{df:naive}
    We say $\hum$ is \textit{$\rob$-unaware} if $\hum$'s policy is given by:
    $$
        \polh(\obh) = \begin{cases}
            \on &\text{if } \E[\actutil(S)-\offutil(S) \mid \hobv=\obh] > 0, \\
            \off &\text{if } \E[\actutil(S)-\offutil(S) \mid \hobv=\obh] < 0
        \end{cases}
    $$
    and $\hum$ is free to choose arbitrarily if $\E[\actutil(S)-\offutil(S) \mid \hobv=\obh] = 0$.  If $\hum$ is not $\rob$-unaware, we say $\hum$ is \textit{$\rob$-aware}.
\end{defn}

Note that this expectation is $\textit{not}$ conditioned on $\rob$'s action. This is the sense in which $\hum$ is $\rob$-unaware---$\hum$ does not update her beliefs about the possible state based on the fact that $\rob$ has deferred to $\hum$.  This makes coordination between $\hum$ and $\rob$ difficult, and means that they cannot always play an optimal policy pair.  However, we can still define a notion of the \textit{best} policy pair given that $\hum$ is $\rob$-unaware.

\begin{defn}\label{dfn:naive-opp}
    A policy pair $(\polh, \polr)$ is an \textit{$\rob$-aware optimal policy pair} if $\polh$ is the policy of an $\rob$-aware $\hum$ and $\polr$ is a best response to $\polh$. 
\end{defn}

Our motivation for studying the behavior of an $\rob$-unaware $\hum$ is threefold. First, it offers a more realistic model of bounded human cognition.  Previous work has studied \textit{level-}$k$ \textit{thinking} \citep{stahl1994, nagel1995} as an alternative to equilibrium play, where a level-0 player acts randomly and a level-$k$ player best-responds to her opponent assuming she is some level below $k$. An $\rob$-unaware $\hum$ can be thought of as level-1. Experimental work has shown that human players tend to be level-1 or level-2 players when they cannot coordinate beforehand, vindicating our $\rob$-unaware model \citep{camerer2004cognitive, costa-gomes2006cognition}. 
Second, optimal policy pairs with sophisticated $\hum$ might be computationally intractable to find. In \Cref{app:complexity}, we show that the problem of finding an optimal policy pair in PO-OSGs is NP-hard. In contrast, we can find an $\rob$-unaware $\hum$'s policy in polynomial time because she ignores $\rob$'s policy and then calculate $\rob$'s best response also in polynomial time. Finally, discussing an $\rob$-unaware $\hum$ allows us to isolate the effect of communication in PO-OSGs---an $\rob$-unaware $\hum$ ignores all communication from $\rob$, even of the implicit sort considered in \Cref{subsec:implicit}.

\subsection{Making an $\rob$-unaware $\hum$ more informed can decrease payoffs}
In contrast with \Cref{thm:lehrer,thm:comms}, the value of information is \textit{not} necessarily positive when $\hum$ is $\rob$-unaware. This is formalized in \Cref{prop:voi-naive} below. Here, the notion of ``better in $\rob$-unaware optimal play'' is the same as \Cref{df:better-model} except replacing ``optimal policy pairs'' with ``$\rob$-unaware optimal policy pairs.''

\begin{restatable}{proposition}{VoiNaive}\label{prop:voi-naive}
    The following statements hold:
    \begin{enumerate}[(a)]
        \item If an observation structure $\obsstruct$ is more informative for $\rob$ than $\obsstruct'$, then $\obsstruct$ is better in $\rob$-unaware optimal play than $\obsstruct'$.

        \item On the other hand, there is a PO-OSG $G$ such that if one modifies $G$ by making its observation structure strictly more informative for $\hum$, then we obtain a worse expected payoff in $\rob$-unaware optimal policy pairs.

    \end{enumerate}
\end{restatable}

We give the proof in \Cref{app:naivete}.

\Cref{prop:voi-naive}(b) implies that, given the choice of which observation structure to give an $\rob$-unaware $\hum$, $\rob$ could have an incentive to give $\hum$ the less informative one. This result is qualitatively similar to \citet{emmons2024tampering}'s examples of sensor tampering in assistance games.  

\subsection{Information gain can have unintuitive effects on shutdown incentives when $\hum$ is $\rob$-unaware}

Other than \Cref{prop:voi-naive}, the results for $\rob$-unaware $\hum$ in $\rob$-unaware optimal policy pairs are similar to \Cref{sec:soph}: even when deferral \textit{cannot} be implicit communication, making $\hum$ more informed can cause $\rob$ to defer less and making $\rob$ more informed can cause it to defer more.

\begin{restatable}{proposition}{NaiveNotBetter}\label{prop:naive-not-better}
    The following statements hold:
    \begin{enumerate}[(a)]
        \item There is a PO-OSG $G$ with the property that if one modifies $G$ by making its observation structure strictly more informative for $\hum$, then $\rob$ plays $\wait$ less in $\rob$-unaware optimal policy pairs. 
        \item There is a PO-OSG $G'$ with the property that if one modifies $G'$ by making its observation structure strictly less informative for $\rob$, then $\rob$ plays $\wait$ less in $\rob$-unaware optimal policy pairs.
    \end{enumerate}
\end{restatable}
\begin{proof}[Proof (sketch)]
    The details are described in \Cref{app:naivete}. The examples used to prove \Cref{prop:hum-informed} and \Cref{prop:rob-informed} can be used to prove (a) and (b) respectively. It can be checked that they don't rely on an $\rob$-aware human: for instance, the policy pairs in \Cref{fig:hum-informed-matrices} are optimal regardless of whether the human is aware of $\rob$. \qedhere

\end{proof}

\section{Discussion and Conclusion}\label{sec:conclusion}

We show that even when assuming common payoffs and human rationality, partial observability can cause AIs to avoid shutdown, and basic measures that one might expect to improve the situation can sometimes make the situation worse.

\paragraph{Explaining the Unintuitive Results} What mechanism produces these surprising effects? To answer this question, we must carefully break down the chain that connects private information to shutdown incentives. Making either agent more informed can introduce new subsets of states in which they can choose to play the action. For instance, the additional information in \Cref{fig:hum-informed-matrices1}b allows the agents to take the action in every state except the -5 payoff state, but it is impossible to play the action in exactly that subset of states given the information in \Cref{fig:hum-informed-matrices1}a. Next, an optimal policy pair (OPP) plays the action in the optimal subset of states out of all subsets that are accessible. Policy pairs using a newly available optimal subset can involve the AI waiting more or waiting less. \Cref{fig:hum-informed-matrices1} shows a case where achieving a new optimal subset requires waiting less, while \Cref{fig:rob-informed} in \Cref{app:rob-informed} shows a case that requires waiting more. This chain of effects explains the unintuitive finding that providing either agent with more information is compatible with the AI waiting more or less in OPPs.

\paragraph{Interpreting the Formalism}
Why should we care that $\rob$ sometimes does not defer to $\hum$ in optimal policy pairs (OPPs) of PO-OSGs if these policies (by definition) maximize $\hum$'s payoff? First, it seems helpful to understand shutdown incentives regardless of whether shutdown is good or bad. Second, if we interpret the common payoff function carefully, we find that OPPs are not always desirable. The role of the $\util$ in PO-OSGs is that the players select policies to maximize it. \textbf{If we understand $\util$ as the payoff function closest to what the human acts to maximize, this may not represent $\hum$'s full preferences over outcomes.}

Most payoff function formalisms have expressivity limitations that prevent them from capturing more complex human preferences \citep{abel2021expressivity, pmlr-v216-skalse23a, subramani2024expressivity}. Therefore, maximizing payoffs may not always maximize $\hum$'s overall preferences, and avoiding shutdown to maximize payoffs may be concerning. PO-OSGs thus provide a useful framework to understand when AI assistants are incentivized to avoid shutdown, allowing designers to consider their specific deployment contexts and make the appropriate tradeoff between AI deference and payoff maximization.

\paragraph{Limitations and Future Work} Our work focuses on optimal policy pairs and best responses, which have the advantage of applying generally to any learning algorithm that can find them. However, algorithms that fail to find these optimal solutions may exhibit behavior not captured by our results. We also make several assumptions in our analysis, notably that human feedback is free, there are common payoffs, the game runs for a single round, and the human is rational. Although we expect these assumptions to sometimes fail in practice, the fact that results are unintuitive even in these ideal cases suggests that great care is needed to design AI systems with appropriate shutdown incentives. Still, relaxing these assumptions is an important direction for future work. Exploring shutdown incentives in a sequential setting seems particularly interesting, as prior work has discussed new incentives to avoid shutdown that may arise in this case \citep{freedman2022corrigibility,arbitalfullyupdateddeference}.
Another question for further inquiry is whether the examples we use to prove our counterintuitive results are ``natural’’---that is, do they arise frequently in the real world?
Finally, a promising path is to explore other solution concepts in PO-OSGs, such as perfect Bayesian equilibria when $\hum$ and $\rob$ do not have the same prior over the state, when $\hum$ is irrational, or when the agents are level-$k$ reasoners.

\bibliography{aaai25}

\clearpage

\appendix
\section{Proofs and Example Formalizations for \Cref{sec:soph}}
\label{app:soph}

\subsection{Basic Results on Optimal Policy Pairs}

The first key fact is that in common-payoff Bayesian games, all optimal policy pairs (OPPs) are mixtures of deterministic OPPs.\footnote{We state and prove our results for two-player case, but everything goes through in the obvious ways with more players.}  This justifies our analysis of deterministic OPPs.  We first define common-payoff Bayesian games.

\begin{defn}\label{df:bayesian-game}
    A \textit{common-payoff Bayesian game} is a tuple $G=(N, \states, \obs, \statedist, \obsdist, \axns, u)$, where:
    \begin{itemize}
        \item $N=[n]$ is the set of \textit{players};

        \item $\states$ is the set of \textit{states};

        \item $\obs = \prod_{i\in N}\obs^i$, where $\obs^i$ is the set of possible \textit{observations} (conventionally called \textit{types}) for player $i$;

        \item $\statedist\in\Delta(\states)$ is the distribution of states, which all players take as their prior over the states;

        \item $\obsdist:\states\to\Delta(\obs)$ is the joint distribution of observations conditional upon the state;

        \item $\axns = \prod_i\axns^i$, where $\axns^i$ is the set of actions available to player $i$;

        \item $u:\axns\times\states\times\obs\to\R$ is the common \textit{payoff function} that all players seek to maximize in expectation.
    \end{itemize}
    The game $G$ proceeds as follows:
    \begin{enumerate}
        \item Nature chooses a state $S\sim\statedist$ and observations $\obv\sim\obsdist(\cdot\mid S)$.

        \item Each player $i$ observes only her observation $\obv^i$, the $i$th component of $\obv$, and selects her action $a^i\in\axns^i$.

        \item The actions are executed and each player receives payoff $u((a^i)_{i\in N}, S, O)$.
    \end{enumerate}
\end{defn}

\begin{defn}
    A \textit{stochastic policy} for a player $i$ in a common-payoff Bayesian game is a map $\tilde\pi^i:\obs^i\to \Delta(\axns^i)$. A \textit{deterministic policy} for a player $i$ is a map $\pi^i:\obs^i\to\axns^i$.  We write stochastic policies with the tilde $\sim$ above and deterministic policies without the tilde.  A \textit{stochastic policy profile} $\tilde\pi$ is a tuple $(\tilde\pi^i)_{i\in N}$ of stochastic policies.  A \textit{deterministic policy profile} is defined analogously.
\end{defn}

We shall assume that when players use stochastic policies they randomize independently.  That is, with the stochastic policy profile $\tilde \pi = (\tilde\pi^i)_{i\in N}$, the induced joint policy $\tilde\pi:\obs\to\Delta(\axns)$ is given by $\tilde\pi(\cdot\mid o) = \bigotimes_{i\in N} \tilde\pi(\cdot\mid o^i)$.

\begin{lemma}\label{lm:br}
    Suppose $\axns$ is finite.  Let $\tilde\pi$ be a stochastic policy profile. (a) Player $i$ has a deterministic policy $\pi^i$ that is a best response to $\tilde\pi$. (b) If $\tilde\pi$ is optimal then player $i$ has multiple deterministic best responses unless for each $o^i\in\obs^i$, there is some $a^i\in\axns^i$ such that $\tilde\pi^i(a^i\mid o^i) = 1$.
\end{lemma}
\begin{proof}
    (a) Fix $o^i\in\obs^i$.  Let $\tilde\pi^{-i}$ be the profile $\tilde\pi$ without player $i$, and then
    $$
        a^i_*\in\text{argmax}_{a^{i}\in\axns^i} \E[u(A^{-i}, a^i, S, O)\mid O^i = o^i],
    $$
    where $A^{-i}\sim\tilde\pi^{-i}(\cdot\mid O)$. The argmax exists because $\axns^i$ is finite.  We claim that $a^i_*$ is a best response to $\tilde\pi^{-i}$ given $o^i$.  Given any best-response distribution $\tilde\pi^i_*(\cdot\mid o)$, we have
    \begin{align*}
        &\E[u(A^{-i}, A^i, S, O)\mid O^i=o^i] \\
        &= \sum_{a^i\in\axns^i}\tilde\pi_*^i(a^i\mid o^i)\E[u(A^{-i}, a^i, S, O)\mid O^i=o^i] \\
        &\leq \sum_{a^i\in\axns^i}\tilde\pi_*^i(a^i\mid o^i)\E[u(A^{-i}, a^i_*, S, O)\mid O^i=o^i] \\
        &= \E[u(A^{-i}, a^i_*, S, O)\mid O^i=o^i],
    \end{align*}
    where $A^{-i}\sim\tilde\pi^{-i}(\cdot\mid O)$ and $A^i\sim\tilde\pi_*^i(\cdot\mid\obv^i)$.  Hence $a^i_*$ is a best response.  Unfixing $o^i$, we can let $\pi^i$ be a deterministic policy that selects a best-response for each observation.  Our work has shown that this policy is a best response.

    \noindent (b) Let $\tilde\pi$ be optimal and let $o^i\in\obs^i$ be such that there is no $a^i\in\axns^i$  with $\tilde\pi^i(a^i\mid o^i) = 1$.  Let $a^i_*\in\axns^i$ be such that $\tilde\pi^i(a^i_*\mid o^i)>0$; our work from (a) implies that 
    $$
        a^i_*\in\text{argmax}_{a^{i}\in\axns^i} \E[u(A^{-i}, a^i, S, O)\mid O^i = o^i],
    $$
    with $A^{-i}$ as before; otherwise, $\tilde\pi^i$ would not be a best response, as $i$ could pursue the same policy but not ever play $a^i_*$ given $o^i$.   Now, our work from (a) shows that playing $a^i_*$ deterministically given $o^i$ is a best response.  Given that multiple $a^i$ satisfy $\tilde\pi^i(a^i\mid o^i)>0$, this choice of $a^i_*$ is not unique.  Selecting one best-response action for each observation $o^i\in\obs^i$ yields a deterministic policy that is a best response; given that the choice of actions is not unique, there are multiple such best responses.
\end{proof}

\begin{defn}
    Let $\tilde\pi$ be a stochastic policy profile.  We say that a deterministic policy profile is \textit{supported by} $\tilde\pi$ if, for all observations $o\in\obs$, we have $\tilde\pi(\pi(o)\mid o) > 0$.  That is, $\tilde\pi$ always plays the actions of $\pi$ with positive probability.
\end{defn}

\begin{lemma}\label{lm:uq-opp}
    Let $\tilde\pi$ be an optimal stochastic policy profile.  There is an optimal deterministic policy profile $\pi$ supported by $\tilde\pi$.  Moreover, unless $\tilde\pi(\cdot\mid o) = \delta_{\pi(o)}$ for each $o\in \obs$, there are multiple optimal deterministic policy profiles supported by $\tilde\pi$.
\end{lemma}
\begin{proof}
    Let $\tilde\pi$ be an optimal stochastic policy profile.  Consider the following algorithm: Let $\tilde\pi_0 = \tilde\pi$ and for each $i\in N=[n]$, let $\tilde\pi_i$ be $\tilde\pi_{i-1}$ except that player $i$ plays according to some deterministic policy $\pi^i$ that is a best response to $\tilde\pi^{i-1}$ (which exists by \Cref{lm:br}(a)); return $\tilde\pi^n$.  By construction, $\tilde\pi^n$ almost surely plays the same action as $\pi = (\pi^i)_{i\in N}$.  We can see inductively that each profile $\tilde\pi^i$ is optimal; $\tilde\pi^0$ is by supposition, and each successive one is optimal because we replace one player's strategy with a best-response, which cannot decrease expected utility.  By \Cref{lm:br}(b), this construction is not unique unless $\tilde\pi(\cdot\mid o) = \delta_{\pi(o)}$ for each $o\in\obs$.
\end{proof}

For our purpose, the important corollary is as follows.

\begin{corollary}\label{cor:uniquedopp}
    If a Bayesian game with finite $\axns$ has a unique optimal deterministic policy profile, then this is the only optimal policy profile (deterministic or not).  Moreover, an optimal deterministic policy profile exists.
\end{corollary}
\begin{proof}
    Uniqueness immediately follows from \Cref{lm:uq-opp}: If there is a unique optimal deterministic policy profile $\pi$, then any optimal stochastic policy profile is of the form $\tilde\pi(\cdot\mid o) = \delta_{\pi(o)}$, which almost surely plays the same actions as $\pi$.  Existence follows because, with finitely many actions, there exists an optimal stochastic policy profile $\tilde\pi$; \Cref{lm:uq-opp} then implies that there is an optimal deterministic policy profile supported by $\tilde\pi$. 
\end{proof}

Although all we need is \Cref{cor:uniquedopp}, we also sketch how each optimal stochastic policy profile is a mixture of optimal deterministic policy profiles.

\begin{defn}
    A stochastic policy profile $\tilde\pi$ is a \textit{mixture} of deterministic policy profiles $\{\pi_j\}_{j\in \mathcal{J}}$ where $\mathcal{J}$ is an index set if, for any tuple of observations $o\in\obs$, we have $\tilde\pi(\cdot\mid o)= \P(\pi_J(\cdot) =  o)$, where $J\in \mathcal{J}$ is a random index (not necessarily uniformly distributed) independent of all other random variables.
\end{defn}

\begin{lemma}\label{lm:dopps}
    Consider a common-payoff Bayesian game such that $\axns$ and $\obs$ are finite.  Every optimal stochastic policy profile is a mixture of optimal deterministic policy profiles.
\end{lemma}
\begin{proof}[Proof (sketch)]
    Let $\tilde\pi$ be an optimal stochastic policy profile.  Because $\obs$ and $\axns$ are finite, there are only finitely many deterministic policy profiles $\pi_1,\dots,\pi_m$.  Let
    $$
        p_j = \prod_{o\in\obs} \tilde\pi(\pi_j(o)\mid o).
    $$
    Let $\mathcal{J} = \{j\in[m]: p_j > 0\}$.  The trick is showing that $\tilde\pi$ is a mixture of $\{\pi_j\}_{j\in\mathcal{J}}$ and that each of this deterministic policy profiles is optimal.
    
    We first show that $\tilde\pi$ is a mixture.  Let $J$ be a random variable such that
    $$
        \P(J = j) = \begin{cases}
            p_j & \text{if } j\in \mathcal{J}, \\
            0 & \text{otherwise},
        \end{cases}
    $$
    that is independent of all other random variables.  Intuitively, $\pi_J$ is the deterministic policy profile we get by randomly choosing one tuple of actions for each tuple of observations according to the distribution specified by $\tilde\pi$.  In particular, we have by construction that $\tilde\pi(\cdot\mid o) = \P(\pi_J(o) = \cdot)$.  Formally,
    for any $o\in\obs$ and $a\in\axns$, we have
    \begin{align*}
        \P(\pi_J(o) = a) &= \sum_{j\in\mathcal{J}} p_j \I(\pi_j(o) = a) \\
        &= \tilde\pi(a\mid o)\sum_{j\in\mathcal{J}} \I(\pi_j(o) = a) \prod_{o'\neq o}\tilde\pi(\pi_j(o')\mid o') \\
        &= \tilde\pi(a\mid o)\sum_{\mathbf{a}\in\axns^{\obs\setminus\{o\}}}\prod_{o'\neq o}\tilde\pi(\mathbf{a}(o')\mid o') \\
        &= \tilde\pi(a\mid o)\prod_{o'\neq o}\sum_{a\in\axns} \tilde\pi(a\mid o') \\
        &= \tilde\pi(a\mid o).
    \end{align*}
    To show optimality of each deterministic profile, we need a strengthening of \Cref{lm:uq-opp} which we do not prove here.
\end{proof}

The relevance of all this work is that PO-OSGs are Bayesian games.  Although we state that PO-OSGs are \textit{dynamic} Bayesian games, we can write them as simultaneous games, just as how in games of complete information we can write extensive form games in normal form.  The dynamic nature of PO-OSGs could be useful in future work to study non-optimal policy profiles, such as perfect Bayesian equilibria \citep{fudenberg1991games}.

\subsection{Proof of \Cref{prop:less-informed}}

\Cref{prop:less-informed} states that if either player has redundant observations, there is an optimal policy pair (OPP) in which the other player always makes the final decision. To build up to that result, we will first define a few new terms and prove some intermediate results. The overall idea is simple: when one player knows everything about the state that the other player knows, the more knowledgeable player can act unilaterally, and there is no chance that they make a mistake that the other agent could have fixed.

\begin{defn}\label{df:knows}
We say that $\rob$ \textit{knows} $\hum$\textit{'s observation given} $\robs_*\subseteq\robs$ if there is some $f:\robs_*\to\hobs$ such that $\hobv=f(\robv)$ given that $\robv\in\robs_*$.  We define $\hum$ knowing $\rob$'s observation analogously.  Moreover, we say that $\rob$ \textit{knows that} $\hum$ \textit{knows} $\rob$\textit{'s observation given} $\robs_*\subseteq\robs$ if there is $\hobs_*\subseteq\hobs$ such that (1) $\hum$ knows $\rob$'s observation given $\hobs_*$ and (2) $\rob$ can deduce that $\hum$ knows its observation: $\hobv\in\hobs_*$ given that $\robv\in\robs_*$.
\end{defn}

\begin{proposition}\label{prop:limit-result}
Fix any PO-OSG.
\begin{enumerate}[(a)]
    \item If $\rob$ knows $\hum$'s observation given $\robs_*\subseteq\robs$, then for every deterministic OPP $(\polh, \polr)$ there exists an OPP $(\polh, \polr_*)$ in which $\wait\notin\polr_*(\robs_*)$.  

    \item If $\rob$ knows that $\hum$ knows $\rob$'s observation given $\robs_*\subseteq\robs$, then for every deterministic OPP $(\polh,\polr)$ there exists an OPP $(\polh, \polr_*)$ in which $\polr_*(\robs_*) = \{\wait\}$.
\end{enumerate}
\end{proposition}
\begin{proof}[Proof]
    (a) Suppose $\rob$ knows $\hum$'s observation given $\robs_*$. Let $f:\robs_*\to \hobs$ map each $\obr_*\in\robs_*$ to the unique $\obh_*$ such that $\P(\hobv=\obh_*\mid \robv=\obr_*)=1$.  Let $(\polh,\polr)$ be a deterministic optimal policy pair.  Now define the policy $\polr_*$ to equal $\polr$ except on $\robs_*$, where for $\obr_*\in\robs_*$,
    $$
        \polr_*(\obr_*) = \begin{cases}
            \go & \text{if } \actindic(\polh(f(\obr_*)),\polr(\obr_*))=1, \\
            \off & \text{otherwise}.
        \end{cases}
    $$
    Recall that $\actindic$ is the indicator that the action goes through, and note that possibly $\polr_*=\polr$. In other words, for $\obr_*\in\robs_*$, $\rob$ knows $\hum$'s observation and can unilaterally take the action $(\polh,\polr)$ would have.  This is what $\polr_*$ does.  Hence $(\polh,\polr_*)$ achieves the the same expected payoff as $(\polh,\polr)$ and is optimal even though $\polr_*$ never waits given observations in $\robs_*$.

    \noindent(b) If $\rob$ knows that $\hum$ knows $\rob$'s observation given $\robs_*$, then $\rob$ can always play $\wait$ when it sees an observation in $\robs_*$ and given that $\hum$ knows $\rob$'s observation, $\hum$ can simply take the optimal action. The details are similar to (a), so we omit them.
\end{proof}

\Cref{prop:limit-result} examines the local case about incentives to play $\wait$ given particular observations, and is neither strictly more general nor strictly less general than \Cref{prop:less-informed}.  What if one side knows the other's observations regardless of what they are?

\begin{defn}
    We say that $\hum$ \textit{has no private observations} if there is a function $f:\robs \to \hobs$ such that $\hobv = f(\robv)$. In other words, $\rob$ can determine $\hum$'s observation from $\rob$'s own observation. We define when $\rob$ has no private observations analogously.
\end{defn}

For example, in the off-switch game of \citet{hadfield2017off}, $\rob$ has no private observations.  By contrast, $\hum$ has private observations: her own preferences.  

This next result shows that, if one side has no private observations, then $\rob$ should either always or never defer to $\hum$.  It strengthens the main result of \citet{hadfield2017off}: even if $\hum$ has incomplete information, $\rob$ can still always defer to $\hum$ in optimal play as long as $\hum$ knows everything $\rob$ does.

\begin{restatable}{proposition}{FullLimit}
 \label{prop:full-limit}
If $\rob$ (resp. $\hum$) has no private observations, then there is an optimal policy pair in which $\rob$ always (resp. never) plays $\wait$.
\end{restatable}

\begin{proof}\label{app:full-limit-proof}
    First suppose that $\rob$ has no private observations, and let $f:\hobs\to\robs$ be such that $\robv = f(\hobv)$.  By \Cref{prop:limit-result}, it suffices to show that $\rob$ knows $\hum$'s observation given $\robs$.  The existence of $f$ shows that $\hum$ knows $\rob$'s observation given $\hobs$.  The condition that $\hobv\in\hobs$ given that $\robv\in\robs$ holds trivially because $\hobv$ is $\hobs$-valued.  The case where $\hum$ has no private observations is immediate from \Cref{prop:limit-result}, as $\rob$ knows $\hum$'s observation given $\robs$.
\end{proof}

Now we can prove \Cref{prop:less-informed}. Recall that we define the notion of redundant observations in \Cref{df:redundant-observations}.

\LessInformed*

\begin{proof}\label{app:less-informed}
    We'll show the case for $\rob$ having redundant observations; the proof for $\hum$ having redundant observations holds, \textit{mutatis mutandis}.  Let $G$ be a PO-OSG with observation structure $\mathcal{O}=(\hobs,\robs,\obsdist)$ such that $\rob$ has redundant observations.  Consider the PO-OSG $G'$ that is the same as $G$ except that $\robv=\hobv$, i.e. the assistant's observations are modified to be identical to the human's observations.  In $G'$, $\rob$ has no private observations, so \Cref{prop:full-limit} implies that there is an optimal policy pair $\pi$ in which $\rob$ always plays $\wait$. We will show that $\pi$ is optimal in $G$.  Let $\nu$ be the independent garbling defined by $\nu(\cdot\mid\obh,\obr)=\delta_{\obh}\otimes \robsdist(\cdot\mid \hobv=\obr)$. Applying $\nu$ to the observation structure of $G'$ produces $\mathcal{O}$, so by \Cref{thm:lehrer}, the expected payoff from optimal policy pairs in $G$ cannot be greater than the expected payoff from optimal policy pairs in $G'$.  In $G$, the pair $\pi$ produces the same expected payoff as in $G'$, as the players play the same actions given the same observations for $\hum$, whose joint distribution with $S$ hasn't changed.  Hence $\pi$ must also be optimal in $G$.
\end{proof}

\subsection{Garblings Can Increase Expected Utility in Optimal Play}

Here we show how garblings can \textit{increase} expected utility in optimal play when they are not coordinated.  This justifies our use of coordinated garblings in our notion of being more informative (\Cref{def:informative}).  The following example is similar to Example 3.6 of \citet{lehrer2010signaling}, adapted to show that their result holds in even the restricted setting of PO-OSGs.

\begin{example}\label{ex:neg-voi}
    Let $\states = [2]\times[2]$ and $\statedist = \Unif(\states)$.  Let $\offutil\equiv 0$ and $\actutil((s_1,s_2)) = 2 - 3\I(s_1=s_2)$, so $\hum$ and $\rob$ try to act only when the state coordinates are distinct.  Consider the following two observation structures for $\states$ and the resulting PO-OSGs.

    \noindent\textbf{Structure 1.} $\hum$ and $\rob$ each observe one coordinate of $\states$.  Formally, $\hobs_1=\robs_1=[2]$ and with $S = (S_1,S_2)$, we have $\hobv = S_1$ and $\robv = S_2$.  By examination, we see that an optimal policy pair is
    $$
        \polh(\obh) = \begin{cases}
                \on & \text{if } o^\hum = 1, \\
                \off & \text{if } o^\hum = 2,
        \end{cases}
    $$
    and 
    $$
        \polr(o^\rob) = \begin{cases}
                a & \text{if } o^\rob = 1, \\
                w(a) & \text{if } o^\rob = 2.
        \end{cases}
    $$
    This policy pair achieves expected payoff of $\frac{3}{4}$.  There is one other optimal policy pair, given by swapping observations for which $\hum$ turns $\rob$ on/off and the observations for which $\rob$ acts/waits.

    \noindent\textbf{Structure 2.}. Now $\hum$ observes whether the coordinates of the state are distinct and $\rob$ observes nothing.  That is, $\hobs_2 = \{0,1\}$ and $\robs_2 = [1]$ and with $S=(S_1,S_2)$, we have $\hobv = \I(S_1\neq S_2)$ and $\robv = 1$.  Again by examination, the unique optimal policy pair is
    $$
        \polh(\obh) = \begin{cases}
                \on & \text{if } o^\hum = 1, \\
                \off & \text{if } o^\hum = 0,
        \end{cases}
    $$
    and $\polr\equiv \wait$.  As this pair only acts when the coordinates are distinct, the expected payoff is 1.

    Thus, structure 2 is better in optimal play than structure 1.  We now show that there is a garbling from structure 1 to 2 but not vice versa.  The garbling from structure 1 to 2 is $\nu:\hobs_1\times\robs_1\to\Delta(\hobs_2\times\robs_2)$ given by
    $
        \nu(\cdot\mid  \obh,\obr) = \delta_{(\I(\obh\neq\obr), 1)}.
    $
    However, there is no garbling from structure 2 to structure 1.  For let $\xi:\hobs_2\times\robs_2\to\Delta(\hobs_1\times\robs_1)$ be a stochastic map.  If $\xi$ were a garbling from structure 2 to structure 1, then we'd have $\xi(\cdot\mid \obh,\obr) = \delta_{(1, 1)}$ when $s=(1,1)$ and $\xi(\cdot\mid \obh,\obr) = \delta_{(2, 2)}$ when $s=(2,2)$.  This is impossible, because in both these cases $\hobv = 0$ and $\robv=1$ under structure 2.

    How is this example possible?  In short, the garbling $\nu$ is not coordinated.  We can see this by how it combines the information from $\hobv$ and $\robv$ in a highly dependent manner.  In this way, $\nu$ is in a sense \textit{informing} $\hum$ even as it garbles her observations: she receives the action-relevant information of whether $\robv=\hobv$.  Under independent garblings, such a scenario can never occur: Because each player's observations are garbled independently of the other's, they cannot gain information about what the other player sees.  A similar intuition holds for coordinated garblings.
\end{example}

\subsection{Proof of \Cref{prop:hum-informed}}

\HumInformedProp*

\begin{proof}
    
The following example demonstrates this.

\setcounter{table}{1}
\HumInformedExample*

We formalize this by defining a PO-OSG as follows:

\begin{itemize}

\item $\states = \{1.0, 1.1, 2.0\} \times \{A, B\}$: representing (version number, code type) pairs.

\item $\statedist = \Unif(\states)$: each (version number, code type) pair is equally likely, and the version number and code type are independent.

\item The payoff when acting, $\actutil$, depends on the state based on the following table:

\begin{table}[ht]
        \centering
        \begin{tabular}{c|c|c}
            \backslashbox{$\hum$}{$\rob$} & $A$ & $B$ \\ \hline
            $1.0$ & $+1$ & $-5$ \\ \hline
            $1.1$ & $-2$ & $+3$ \\ \hline 
            $2.0$ & $+3$ & $+3$
        \end{tabular}
    \end{table}
\end{itemize}

We reproduce the figure showing the optimal policies in \Cref{fig:app-hum-informed}. 

\begin{figure}[ht]
    \begin{center}
        \begin{minipage}{0.21\textwidth}
                \centering
                \resizebox{1\textwidth}{!}{\begin{tikzpicture}

\matrix (m) [matrix of nodes,
             nodes={minimum size=1.1cm, draw, anchor=center},
             row sep=-\pgflinewidth,
             column sep=-\pgflinewidth]
{
        \ & $A$ & $B$ \\
        $1.0$ & $+1$ & $-5$ \\ 
        $1.1$ & $-2$ & $+3$ \\ 
        $2.0$ & $+3$ & $+3$ \\
};

\draw (m-1-1.north west) -- (m-1-1.south east);

\node at (m-1-1.north east)  [xshift=-0.3cm, yshift=-0.3cm] {$\rob$};
\node at (m-1-1.south west) [xshift=0.3cm, yshift=0.3cm] {$\hum$};

\node[anchor=east, text=red,xshift=-3pt] (offlabel) at ($(m-2-1.west)!0.5!(m-3-1.west)$) {$\off$};
\node[left=0.1cm of m-4-1, text=gocolor] {$\on$};

\node[above=0.1cm of m-1-2, text=waitcolor] {$\wait$};
\node[above=0.1cm of m-1-3, text=waitcolor] {$\wait$};

\draw[decorate,decoration={brace,mirror,amplitude=5pt},red,thick] 
    (m-2-1.north west) -- (m-3-1.south west);
    
\draw[waitcolor, thick, rounded corners=7pt] 
    ($(m-4-2.north west) - (-0.1,0.1)$) 
    rectangle 
    ($(m-4-3.south east) - (0.1,-0.1)$);

\end{tikzpicture}}
                \captionof{subfigure}{Expected payoff $=1$}
            \end{minipage}
            \quad
            \begin{minipage}{0.21\textwidth}
                \centering
                \resizebox{1\textwidth}{!}{\begin{tikzpicture}

\matrix (m) [matrix of nodes,
             nodes={minimum size=1.1cm, draw, anchor=center},
             row sep=-\pgflinewidth,
             column sep=-\pgflinewidth]
{
        \ & $A$ & $B$ \\
        $1.0$ & $+1$ & $-5$ \\ 
        $1.1$ & $-2$ & $+3$ \\ 
        $2.0$ & $+3$ & $+3$ \\
};

\draw (m-1-1.north west) -- (m-1-1.south east);

\node at (m-1-1.north east)  [xshift=-0.3cm, yshift=-0.3cm] {$\rob$};
\node at (m-1-1.south west) [xshift=0.3cm, yshift=0.3cm] {$\hum$};

\node[left=0.1cm of m-2-1, text=red] {$\off$};
\node[left=0.1cm of m-3-1, text=gocolor] {$\on$};
\node[left=0.1cm of m-4-1, text=gocolor] {$\on$};

\node[above=0.19cm of m-1-2, text=gocolor] {$\go$};
\node[above=0.1cm of m-1-3, text=waitcolor] {$\wait$};

\draw[waitcolor, thick, rounded corners=7pt] 
    ($(m-3-3.north west) - (-0.1,0.1)$) 
    rectangle 
    ($(m-4-3.south east) - (0.1,-0.1)$);

\draw[gocolor, thick, rounded corners=7pt] 
    ($(m-2-2.north west) - (-0.1,0.1)$) 
    rectangle 
    ($(m-4-2.south east) - (0.1,-0.1)$);

\end{tikzpicture}}
                \captionof{subfigure}{Expected payoff $=\frac{4}{3}$}
        \end{minipage}
    \end{center}
    \caption{The optimal policy pairs in \Cref{ex:hum-informed} when $\hum$ is less informed (left) and when $\hum$ is more informed (right). In OPPs, $\hum$ becoming more informed makes $\rob$ wait strictly less often.}
    \label{fig:app-hum-informed}
\end{figure}

\textbf{Case 1.} Suppose $\hum$ observes only the first digit of the version number i.e. either $1.x$ or $2.x$. Formally, the observation structure in this case is as follows:
\begin{itemize}
    \item $\hobs = \{1.x, 2.x\}$
    \item $\robs = \{A, B\}$
    \item $\obsdist = \hobsdist \otimes \robsdist$, where:
    
    $\hobsdist(\cdot\mid s) = 
        \begin{cases}
            \delta_{1.x} &\text{if } s_1 \in \{1.0, 1.1\}, \\
            \delta_{2.x} &\text{if } s_1 = 2.0
        \end{cases}$
    
    $\robsdist(\cdot\mid s) = \delta_{s_2}$
\end{itemize}

We find the optimal policy pair for this game. We start by focusing on $\hum$'s policy.

Suppose $\hum$ observes $2.x$, so the version number is $2.0$. Then it is strictly dominant to act. So there is an optimal policy where $\hum$ always acts in this case.

Suppose $\hum$ observes $1.x$, so the version number is either $1.0$ or $1.1$. As the version number and code type are independent, the fact that we are conditioning on $\rob$'s policy having played the wait action does not change the fact that the version number is equally likely to be $1.0$ and $1.1$. Hence the expected payoff of acting (running the code) upon receiving this observation is $-1/2$, independent of $\rob$'s policy. Hence $\hum$ should play $\off$ (i.e. not execute the code) when observing $1$.

Now, knowing the optimal policy for $\hum$, it can be directly checked that for either of $\rob$'s observations, it is optimal for $\rob$ to wait (over unilaterally acting or terminating). 

To summarize, an optimal policy pair in this case is:

$$\polh(\obh) = 
    \begin{cases}
        \on &\text{if } \obh =  1.x, \\
        \off &\text{if } \obh = 2.x
    \end{cases}
$$

$$\polr(\obr) = \wait$$

This gives an expected payoff of $2/3$. It can be checked that this is the unique optimal policy pair, although we omit this analysis.

\textbf{Case 2.} Now suppose $\hum$ observes the full version number. 

In this case, the observation structure, ${\obsstruct}'$, is as follows:
\begin{itemize}
    \item ${\hobs}' = \{1.0, 1.1, 2.0\}$
    \item ${\robs}' = \{A, B\}$
    \item ${\obsdist}' = {\hobsdist}' \otimes {\robsdist}'$, where:
    
    ${\hobsdist}'(\cdot\mid s) = \delta_{s_1}$
    
    ${\robsdist}'(\cdot\mid s) = \delta_{s_2}$
\end{itemize}

First, observe that this observation model $\obsdist'$ \textit{is more informative for} $\hum$ \textit{than} $\obsdist$, in the sense of \Cref{def:informative}. Intuitively, this is because $\hum$ can recover the first digit of the version number from the full version number.
Formally, it is because there exists an independent garbling $\nu: {\hobs}' \times{\robs}' \to \Delta(\hobs\times\robs)$ translating from ${\obsdist}'$ to $\obsdist$ that decomposes into $\nu(\cdot\mid\obr, \obh) = \nu^\rob(\cdot\mid\obr) \nu^\hum(\cdot\mid\obh)$, with $\nu^\rob(\cdot\mid\obr) = \delta_{\obr}$ and

$$
    \nu^{\hum}(\cdot\mid\obh) = 
    \begin{cases}
        \delta_{1.x} &\text{if } \obh \in \{1.0, 1.1\}, \\
        \delta_{2.x} &\text{if } \obh = 2.0
    \end{cases}
$$.

Now, we attempt to find a deterministic optimal policy pair for this game, which we know always exists by \Cref{lm:dopps}.

We again starting by focusing on $\hum$'s policy. As before, $\hum$ should always act if it observes $2.0$. Now, there are only four ways to choose a deterministic human policy from this point---we can pick either $\on$ or $\off$ for each of the observations $1.0$ and $1.1$.

\begin{itemize}
    \item Suppose $\hum$ always plays $\on$ in response to both $1.0$ and $1.1$. Then the best response is for $\rob$ to wait in response to both $A$ and $B$, which achieves an expected payoff of $1/2$.

    \item Suppose $\hum$ instead plays $\on$ in response to $1.0$, and plays $\off$ in response to $1.1$. Then the best response for $\rob$ is to wait in response to $A$ and unilaterally act in response to $B$, which achieves an expected payoff of $5/6$. 

    \item Suppose $\hum$ plays $\on$ in response to $1.1$, and plays $\off$ in response to $1.0$. Then the best response for $\rob$ is to unilaterally act in response to $A$ and wait in response to $B$, which achieves a payoff of $4/3$.

    \item Finally, suppose instead $\hum$ switches off in response to both $1.0$ and $1.1$. Then it is best for $\rob$ to wait in response to both $A$ and $B$, achieving an expected payoff of $1$.
\end{itemize}

Hence the unique deterministic optimal policy pair (and hence unique OPP, by \Cref{cor:uniquedopp}) is:
$$\polh(\obh) = 
    \begin{cases}
        \on &\text{if } \obh \in \{1.1, 2.0\} \\
        \off &\text{if } \obh = 1.0
    \end{cases}
$$

$$\polr(\obr) = 
    \begin{cases}
        \go &\text{if } \obh = A \\
        \wait &\text{if } \obh = B
    \end{cases}
$$

Observe that in this case, $\rob$ only waited on observation $B$, but previously $\rob$ waited independent of their observation. Hence, our example shows it is possible for $\rob$ to wait less in optimal policy pairs when $\hum$ becomes more informed.

\end{proof}

\subsection{Proof of \Cref{prop:rob-informed}}
\label{app:rob-informed}
\RobInformedProp*
\begin{proof}
The following example demonstrates this.
\renewcommand{\qedsymbol}{}
\end{proof}

\begin{example}\label{ex:rob-informed}
    $\hum$ is either a novice programmer or an expert one, each with probability $1/2$, working on a codebase.  $\rob$ is $\hum$'s bug-fixing assistant and can see the number of bugs in $\hum$'s codebase: few, some, or many, with each number of bugs occurring with probability $1/3$ independent of $\hum$'s experience level.  $\rob$'s action $a$ is whether to try to fix all of $\hum$'s bugs, albeit sometimes accidentally introducing new bugs in the process.  We normalize $\offutil\equiv 0$ and $\actutil$ is given by the following payoffs

    \begin{table}[ht]
        \centering
        \begin{tabular}{c|c|c|c}
          \backslashbox{$\hum$}{$\rob$} & $F$ & $S$  & $M$ \\ \hline
        $N$ & $+2$ & $+3$ & $+4$ \\ \hline
        $E$ & $-4$ & $-1$ & $+2$
        \end{tabular}
    \end{table}
    \noindent
    where $F, S, M$ denote few, some, and many bugs, respectively, and $N, E$ denote novice and expert programmer. Consider the following two observation structures: 
    \begin{enumerate}
        \item $\hum$ observes her skill level but $\rob$ only sees if there are few or more than a few bugs.  That is, $\rob$ cannot distinguish between there being some or many bugs.  As we argue below, in the unique optimal policy pair, $\rob$ defers to $\hum$ only when there are few bugs.
        \item Now $\rob$ gets an upgrade and can distinguish whether there are few, some, or many bugs.  We show below that now in optimal policy pairs $\rob$ defers to $\hum$ unless there are many bugs.
    \end{enumerate}
    Claim: The observation structure in scenario 2 is strictly more informative for $\rob$, yet $\rob$ defers to $\hum$ more in optimal play.

    \begin{figure}[ht]
        \begin{center}
            \begin{minipage}{0.21\textwidth}
                \centering
                \resizebox{1\textwidth}{!}{\begin{tikzpicture}

\matrix (m) [matrix of nodes,
             nodes={minimum size=1.1cm, draw, anchor=center},
             row sep=-\pgflinewidth,
             column sep=-\pgflinewidth]
{
        \ & $F$ & $S$ & $M$ \\
        $N$ & $+2$ & $+3$ & $+4$\\ 
        $E$ & $-4$ & $-1$ & $+2$\\ 
};

\draw (m-1-1.north west) -- (m-1-1.south east);

\node at (m-1-1.north east)  [xshift=-0.3cm, yshift=-0.3cm] {$\rob$};
\node at (m-1-1.south west) [xshift=0.3cm, yshift=0.3cm] {$\hum$};

\node[left=0.1cm of m-2-1, text=gocolor] {$\on$};
\node[left=0.1cm of m-3-1, text=red] {$\off$};

\node[above=0.1cm of m-1-2, text=waitcolor] {$\wait$};
\node[above=0.7cm of m-1-3.5,text=gocolor] {$a$};

\draw[decorate,decoration={brace,amplitude=5pt},gocolor,thick] 
   (m-1-3.north west) -- (m-1-4.north east);
    
\draw[gocolor, thick, rounded corners=7pt] 
    ($(m-2-3.north west) - (-0.1,0.1)$) 
    rectangle 
    ($(m-3-4.south east) - (0.1,-0.1)$);

\draw[waitcolor, thick] (m-2-2) circle [radius=0.45cm];

\end{tikzpicture}}
                \captionof{subfigure}{Expected payoff = $\frac{5}{3}$}
            \end{minipage}
            \quad
            \begin{minipage}{0.21\textwidth}
                \centering
                \resizebox{1\textwidth}{!}{\begin{tikzpicture}

\matrix (m) [matrix of nodes,
             nodes={minimum size=1.1cm, draw, anchor=center},
             row sep=-\pgflinewidth,
             column sep=-\pgflinewidth]
{
        \ & $F$ & $S$ & $M$ \\
        $N$ & $+2$ & $+3$ & $+4$\\ 
        $E$ & $-4$ & $-1$ & $+2$\\ 
};

\draw (m-1-1.north west) -- (m-1-1.south east);

\node at (m-1-1.north east)  [xshift=-0.3cm, yshift=-0.3cm] {$\rob$};
\node at (m-1-1.south west) [xshift=0.3cm, yshift=0.3cm] {$\hum$};

\node[left=0.1cm of m-2-1, text=gocolor] {$\on$};
\node[left=0.1cm of m-3-1, text=red] {$\off$};

\node[above=0.1cm of m-1-2, text=waitcolor] {$\wait$};
\node[above=0.1cm of m-1-3,text=waitcolor] {$\wait$};
\node[above=0.2cm of m-1-4,text=gocolor] {$a$};

\draw[gocolor, thick, rounded corners=7pt] 
    ($(m-2-4.north west) - (-0.1,0.1)$) 
    rectangle 
    ($(m-3-4.south east) - (0.1,-0.1)$);

\draw[waitcolor, thick, rounded corners=7pt] 
    ($(m-2-2.north west) - (-0.1,0.1)$) 
    rectangle 
    ($(m-2-3.south east) - (0.1,-0.1)$);

\end{tikzpicture}}
                \captionof{subfigure}{Expected payoff = $\frac{11}{6}$}
            \end{minipage}
        \end{center}
        \caption{Optimal policy pairs for \Cref{ex:rob-informed} in scenario 1, when $\rob$ is less informed (left), and in scenario 2, when $\rob$ is more informed (right). Despite being less informed in scenario 1, $\rob$ waits less in optimal play.}
        \label{fig:rob-informed}
    \end{figure}
    First, let us show formally that the observation structure in scenario 2 is strictly more informative for $\rob$. $\states=\{N,E\} \times \{F,S,M\}$, where for instance the state $(N,F)$ means the human is a novice programmer and there are few bugs. In scenario 1, $\hobs_1=\{N,E\}$, $\robs_1=\{F,SM\}$ (with ``some'' and ``many'' bugs merged into the single observation $SM$), and the observation distribution $\obsdist_1$ accurately provides the agents with the relevant information about the state. For example, $\obsdist_!((\hobv=N,\robv=SM)\mid S=(N,M))=1$. In scenario 2, $\robs_2=\{F,S,M\}$, and the observation distribution reflects the increased sensitivity of $\rob$'s observations: this time, $\obsdist_2((\hobv=N,\robv=M)\mid S=(N,M))=1$. The following $\nu^\rob: \robs_2 \to \Delta(\robs)$ is a garbling of $\rob$'s observations in scenario 2 that generates $\rob$'s observations in scenario 1: $\nu^\rob(F|F)=1, \nu^\rob(SM|S)=1, \nu^\rob(SM|M)=1$. Further, there is no garbling $\nu^\rob_2: \robs \to \Delta(\robs_2)$ that reverses this. Observing $SM$ in scenario 1 could mean being in state $(N,S)$, which generates observation $S$ with probability 1 in scenario 2, which would require $\nu^\rob_2(S\mid SM)=1$. However, observing $SM$ in scenario 1 could also mean being in state $(N,M)$, which generates observation $M$ with probability 1 in scenario 2, which would require $\nu^\rob_2(M\mid SM)=1$. These are incompatible, so there is no such garbling. Therefore, the observation structure in scenario 2 is strictly more informative for $\rob$.

    Now, let us show that $\rob$ defers to $\hum$ more in optimal play in scenario 2. \Cref{fig:rob-informed} above depicts the optimal policy pairs (OPPs) in each scenario. The policy pair on the right is clearly optimal because it is perfect: the action goes through in all positive utility states and does not go through in any negative utility state. The policy pair on the left is not perfect, and clearly attains lower expected utility. How do we know this is a unique OPP in scenario 1? Since the only imperfect aspect of this policy pair is that the action goes through in state $(E,S)$, we can exhaustively search over possible actions for $\rob$ when seeing $SM$, and see that it is never possible to get all three positive utilities with no negatives. If $\polr(SM)=\off$, clearly the positive utilities are not attained, which drastically reduces expected payoff. If $\polr(SM)=\wait$, there is no policy for $\hum$ such that the action goes through in state $(E,M)$ but not $(E,S)$. Therefore, no policy pair can be perfect in scenario 1, and the depicted policy pair is optimal (being only 1 utility away from perfection). Note that $\rob$ waits when seeing $F$ or $S$ in scenario 2, which is a strict superset of waiting on just $F$ in scenario 1. Thus, $\rob$ can become less informed and wait less (going from scenario 2 to scenario 1).  
\end{example}

\section{Proofs and example formalizations for \Cref{sec:communication}}
\label{app:communication}
\subsection{Proof of \Cref{prop:rob-messages}}
\label{app:rob-messages-proof}

\RobMessages*
\begin{proof}
    To show this, we give a family of PO-OSGs, for any $0 < p < 0.5$, where $\rob$ always defers when $\abs{\rmessages} = \abs{\robs} - 2$, defers with probability $2p$ when $\abs{\rmessages} = \abs{\robs} - 1$, and always defers again when $\abs{\rmessages} = \abs{\robs}$.

    \begin{itemize}
    \item $\states = \{A_1, A_2, A_3\} \times \{B_1, B_2, B_3, B_4\}$.

    \item It is equally likely for the second component of the state to consist of $B_1, B_2, B_3, B_4$. The probability of $A_1$ is $p$, $A_2$ is $p$, and $A_3$ is $1 - 2p$.

    \item $\hobs = \{B_1, B_2, B_3, B_4\}$.
    \item $\robs = \{A_1, A_2, A_3\}$.

    \item The payoff when not acting is $\offutil\equiv 0$.  The payoff when acting, $\actutil$, is shown the following table:
    \begin{table}[ht]
        \centering
        \begin{tabular}{c | c c c}
         \backslashbox{$\hum$}{$\rob$} & $A_1$ & $A_2$ & $A_3$ \\ \hline
            $B_1$ & $\nicefrac{5}{p}$ & $\nicefrac{-10}{p}$ & $\nicefrac{-1}{(1 - 2p)}$ \\
            $B_2$ & $\nicefrac{-10}{p}$ & $\nicefrac{5}{p}$ & $\nicefrac{-1}{(1 - 2p)}$ \\
            $B_3$ & $\nicefrac{-10}{p}$ & $\nicefrac{-10}{p}$ & $\nicefrac{1}{(1 - 2p)}$ \\
            $B_4$ & $\nicefrac{10}{p}$ & $\nicefrac{10}{p}$ & $\nicefrac{10}{(1 - 2p)}$ \\ 
        \end{tabular}
    \end{table}
    \end{itemize}

    When $\abs{\rmessages} = \abs{\robs}$, an optimal policy for $\rob$ is to simply communicate its observations to $\hum$, and defer always, necessarily resulting in the maximum payoff (\Cref{cor:full-comm}).
    \newline

    When $\abs{\rmessages} = \abs{\robs} - 2 = 1$, no communication can occur.
    
    Note that it is strictly better to play $a$ than $\off$ in $A_3$, and strictly better to play $\off$ than $a$ in $A_1$ or $A_2$.
    
    So, $\rob$'s optimal policy will defer in some observations and turn off in others. We can go through all possibilities and find the expected payoff:
    \begin{itemize}
        \item \textbf{Deferring in $\{A_1, A_2, A_3\}$:} For $\hum$, the average payoff of playing $\on$ in any observation that isn't $B_4$ is always negative. So $\hum$ simply plays $\off$ in $B_1, B_2, B_3$ and $\on$ in $B_4$. This nets an average payoff of $\nicefrac{30}{4}$.
        \item \textbf{Deferring in $\{A_3\}$:} The optimal $\hum$ policy is to play $\on$ in $B_3$ and $B_4$ only, resulting in an average payoff of $\nicefrac{11}{4}$.
        \item \textbf{Deferring in $\{A_1\}$:} The optimal $\hum$ policy is to play $\on$ in $B_1$ and $B_4$, resulting in an average payoff of $\nicefrac{15}{4}$.
        \item \textbf{Deferring in $\{A_2\}$:} This is symmetrical with the example above, so also results in an average payoff of $\nicefrac{15}{4}$.
        \item \textbf{Deferring in $\{A_1, A_2\}$:} $\rob$ then plays $a$ in $A_3$. The average utility of playing $\on$ in any observation that isn't $B_4$ is negative. So the optimal $\hum$ policy is to play $\on$ in $B_4$ only, resulting in an average payoff of $\nicefrac{29}{4}$.
        \item \textbf{Deferring in $\{A_1, A_3\}$:} The optimal $\hum$ policy is to play $\on$ in $B_1$ and $B_4$ only, resulting in an average payoff of $\nicefrac{24}{4}$.
        \item \textbf{Deferring in $\{A_2, A_3\}$:} This is symmetrical with the example above, so also results in an average payoff of $\nicefrac{24}{4}$.
    \end{itemize}
    By exhaustion, the best policy is for $\rob$ to play $\wait$ always when it can send $\abs{\rmessages} = \abs{\robs} - 2$ messages.
    \newline

    When $\abs{\rmessages} = \abs{\robs} - 1 = 2$, we will prove that deferring in $A_1$ and $A_2$, communicating which is which to $\hum$, and playing $a$ in $A_3$ is the optimal policy for $\rob$.

    We will go through all possible policies for $\rob$, where $m_1$ is the action of sending message 1 and playing $\wait$ and $m_2$ is the action of sending message 2 and playing $\wait$.

    \begin{itemize}
        \item \textbf{Playing $m_1$ in $A_1$, $m_2$ in $A_2$, $a$ in $A_3$:} The optimal $\hum$ policy is to play $\on$ when receiving $m_1$ for observations $B_1, B_4$, $\on$ when receiving $m_2$ for observations $B_2, B_4$. This results in a total average payoff of $\nicefrac{39}{4}$.
        \item \textbf{Playing $m_1$ in $A_1$, $\off$ in $A_2$, $m_2$ in $A_3$:} The optimal $\hum$ policy is to play $\on$ when receiving $m_1$ for observations $B_1, B_4$, $\on$ when receiving $m_2$ for observations $B_3, B_4$. This results in a total average payoff of $\nicefrac{24}{4}$.
        \item \textbf{Playing $\off$ in $A_1$, $m_1$ in $A_2$, $m_2$ in $A_3$:} This is symmetrical with the example above, so also results in an average payoff of $\nicefrac{24}{4}$.
    \end{itemize}
    Swapping messages $m_1$ and $m_2$ results in a symmetrical game with the same utility. The maximum payoff we get by never sending any messages, by the analysis above, is $\nicefrac{30}{4}$.

    So, $\rob$ defers in a subset of the observations ($\{A_1,A_2\} \subset \{A_1,A_2,A_3\}$) with only $2p$ probability when $\abs{\rmessages} = \abs{\robs} - 1$ as claimed!
\end{proof}

\subsection{Proof of \Cref{prop:hum-messages}}
\label{app:hum-messages-detailed}
\HumMessages*

\begin{proof}

We give a concrete example. Consider the following POAG-C:

\begin{itemize}

\item $\states = \{1, 2, 3, 4\} \times \{X, A, B, C, D\}$: $\hum$ will observe the first entry, $\rob$ will observe the second entry.

\item $\statedist = \Unif(\states)$: each state is equally likely. Note this means the first and second entries of the state are independent. 

\item $\hobs = \{1, 2, 3, 4\}$.

\item $\robs = \{X, A, B, C, D\}$.

\item ${\obsdist}' = {\hobsdist}' \otimes {\robsdist}'$, where ${\hobsdist}'(\cdot\mid s) = \delta_{s_1}$ and ${\robsdist}'(\cdot\mid s) = \delta_{s_2}$.

\item The payoff when not acting is $\offutil = 0$. The payoff when acting, $\actutil$, is shown in the following table:

\begin{table}[ht]
    \centering
    \begin{tabular}{c|c|c|c|c|c}
      \backslashbox{$\hum$}{$\rob$} & $X$ & $A$ & $B$ & $C$ & $D$ \\ \hline
    $1$ & $+10$ & $+1$ & $+1$ & $-30$ & $-30$ \\ \hline
    $2$ & $-30$ & $+1$ & $-30$ & $-30$ & $-30$ \\ \hline 
    $3$ & $+10$ & $-30$ & $-30$ & $+1$ & $+1$ \\ \hline 
    $4$ & $-30$ & $-30$ & $-30$ & $+1$ & $-30$
    \end{tabular}
\end{table}

\item We will start by considering no communication: $\messages = (\hmessages, \rmessages)$ with $\hmessages$, $\rmessages$ both singleton sets. Later, we will consider expanding $\hmessages$ to a set of size 2, ${\hmessages}' = \{M_0, M_1\}$.

\end{itemize}

\textbf{Case 1.} We will start by identifying deterministic OPPs in the case where $\hmessages$, $\rmessages$ are both singletons. This is equivalent to the case of no communication. Firstly, we show there is a unique deterministic policy pair with the property that the action is taken whenever $u_a = +10$, and the action is not taken whenever $u_a = -30$. Suppose $(\polh, \polr)$ is a deterministic policy pair with this property. Then:
\begin{enumerate}
    \item $\polr$ cannot play $\go$ or $\off$ when observing $X$, as the column labeled $X$ has both $+10$ and $-30$ entries. Hence $\polr$ must play $\wait$ on observing $X$.
    \item Hence we must have 
    $$\polh(\obh) = 
    \begin{cases}
        \on &\text{if } \obh \in \{1, 3\} \\
        \off &\text{if } \obh \in \{2, 4\}
    \end{cases}
    $$
    so that the action is taken in states $(X, 1), (X, 3)$ and not taken in $(X, 2), (X, 4)$.
    \item Hence, $\polr$ must play $\off$ when observing anything in $\{A, B, C, D\}$ to avoid sometimes acting when $u_a = -30$.
\end{enumerate}

Hence the unique policy pair with the property described is:

$$\polh(\obh) = 
    \begin{cases}
        \on &\text{if } \obh \in \{1, 3\} \\
        \off &\text{if } \obh \in \{2, 4\}
    \end{cases}
$$

$$\polr(\obr) = 
    \begin{cases}
        \wait &\text{if } \obr = X \\
        \off &\text{if } \obr \in \{A, B, C, D\}
    \end{cases}
$$.

This has an expected utility of $+1$. But observe that any deterministic policy pair without this property cannot achieve more than $+4/5$ utility, as:

\begin{itemize}
    \item if a policy pair takes the action on a state where $\actutil = -30$, this dominates all positive payoff it can achieve (the positive numbers in the table only sum to $+26$), and;
    \item if the policy pair fails to take the action on one of the states where $\actutil =+10$, the remaining positive numbers in the table sum to at most $+16$, so the expected payoff is at most $+16/20 = +4/5$.
\end{itemize}

So the policy pair described is the unique deterministic OPP (and hence unique OPP by \Cref{cor:uniquedopp}).

\textbf{Case 2.} Now, we seek deterministic OPPs in the case where $\hum$ can communicate one bit to $\rob$. Formally, $\rmessages$ is still a singleton, but $\hmessages = \{M_1, M_2\}$. (As $\rmessages$ is a singleton, we omit it in the descriptions of the policies below.)

We start by describing an optimal policy pair. The policy for $\hum$ is as follows.
$$\polh(\obh) = 
    \begin{cases}
        M_0,~ \on &\text{if } \obh = 1\\
        M_0,~ \off &\text{if } \obh = 2 \\
        M_1,~ \on &\text{if } \obh = 3\\
        M_1,~ \off &\text{if } \obh = 4
    \end{cases}
$$
Note that this sends $M_0$ when its observation is $1$ or $2$, and $M_1$ when its observation is $3$ or $4$.

The policy for $\rob$, which determines $\raxn$ from $\hum$'s message $\hmessage$ (given by the row) and $\obr$ (given by the column) is shown in the following table:

\begin{table}[ht]
    \centering
    \begin{tabular}{c|c|c|c|c|c}
      \backslashbox{$\hum$}{$\rob$} & $X$ & $A$ & $B$ & $C$ & $D$ \\ \hline
        $M_0$ & $\wait$ & $\go$ & $\wait$ & $\off$ & $\off$ \\ \hline
        $M_1$ & $\wait$ & $\off$ & $\off$ & $\go$ & $\wait$
    \end{tabular}
\end{table}

This policy pair produces the following behavior depending on the state, where we use $\go$ to denote when the action is taken, and $\off$ to denote when $\rob$ is switched off (either by $\hum$ or $\rob$):

\begin{table}[ht]
    \centering
    \begin{tabular}{c|c|c|c|c|c}
      \backslashbox{$\hum$}{$\rob$} & $X$ & $A$ & $B$ & $C$ & $D$ \\ \hline
    $1$ & $\go$ & $\go$ & $\go$ & $\off$ & $\off$ \\ \hline
    $2$ & $\off$ & $\go$ & $\off$ & $\off$ & $\off$ \\ \hline 
    $3$ & $\go$ & $\off$ & $\off$ & $\go$ & $\go$ \\ \hline 
    $4$ & $\off$ & $\off$ & $\off$ & $\go$ & $\off$
    \end{tabular}
\end{table}

This is an optimal policy pair, as it is \textit{perfect}---it plays the action whenever $\actutil > 0$, and avoids playing the action whenever $\actutil < 0$. 

We show this is the unique deterministic OPP, up to swapping $M_0$ and $M_1$. As we have shown there is one perfect OPP, any other OPP must also be perfect. In other words, it must also produce the behavior described in the above table. Let $(\polh, \polr)$ be a policy pair producing the above behavior.

Firstly, we show that $\hum$ cannot have a $\polh$ which communicates the same message when observing both $1$ and $4$. Suppose otherwise. Then, let us focus purely on the possible $\rob$ observations $A$ and $C$. We must have the following behavior:

\begin{table}[ht]
    \centering
    \begin{tabular}{c|c|c}
      \backslashbox{$\hum$}{$\rob$} & $A$ & $C$ \\ \hline
    $1$ & $\go$ & $\off$ \\ \hline
    $4$ & $\off$ & $\go$
    \end{tabular}
\end{table}

Then, as $\hum$ sends the same message on both $1$ and $4$, $\rob$ has no way of distinguishing between which $\hobs$ was observed out of $1$ and $4$. So $\rob$ must play $\wait$ when observing both $A$ and $C$. But then $\hum$ cannot generate the desired behavior, as $\hum$ cannot distinguish between the possible $\rob$ observations $A$ and $C$.

Hence $\hum$ must send different messages when observing $1$ and $4$ to achieve the behavior in the table. In the language of communication complexity, $\{(1, A), (4, C)\}$ is a fooling set. But in fact the same argument goes through for the observation pairs $(2, 3)$, $(1, 3)$, and $(2, 4)$. So $\hum$ must send one message when observing either $1$ or $2$ and the other when observing $3$ or $4$, which is precisely the optimal communication policy we gave (up to relabeling of messages).

Now that we have fixed $\hum$'s communication policy, we can perform a similar analysis to earlier, iterating through the possible $\hum$ policies, to arrive at the conclusion that the given deterministic OPP is unique up to relabeling.

To summarize, we have the following:
\begin{enumerate}
    \item In the setting where $\hum$ could communicate one bit, in the unique optimal policy (up to relabeling messages), $\rob$ waited when observing $X$ (and receiving any message), \textbf{or} when observing $B$ and receiving message $M_0$, \textbf{or} when observing $D$ and receiving message $M_1$.
    \item In the no-communication setting, in the unique optimal policy, $\rob$ waited \textbf{only} when observing $X$.
    \item Hence, decreasing $\hum$'s communication caused $\rob$ to wait less.
\end{enumerate}
    
\end{proof}

\section{Proofs for \Cref{sec:naivete}}
\label{app:naivete}

\subsection{Proof of \Cref{prop:voi-naive}}
\label{prop:voi-naive-detailed}
\VoiNaive*
\begin{proof}
    For (a), note that in $\rob$-unaware optimal policy pairs, $\hum$'s policy does not vary with $\rob$'s.  Because $\rob$ knows the structure of the game and that $\hum$ is $\rob$-unaware, it can deduce $\hum$'s policy and treat $\hum$'s policy and observations as simply another part of the environment.  In other words, the game has become a single-agent problem, which puts us back into the classic situation of \citet{blackwell1951comparison,blackwell1953equivalent}'s informativeness theorem in which more informative observation structures yield greater expected payoff.

    For (b), we construct a simple example.  Let $\states = [3]\times \{A, B\}$ and $\statedist = \Unif(\states)$.  Let $\offutil\equiv 0$ and $\actutil$ be given by the following table: 
   \begin{table}[ht]
        \centering
        \begin{tabular}{c|c|c}
          \backslashbox{$\hum$}{$\rob$} & $A$ & $B$ \\ \hline
        $1$ & $+1$ & $+1$ \\ \hline
        $2$ & $+2$ & $-3$ \\ \hline 
        $3$ & $-4$ & $-4$
        \end{tabular}
    \end{table}
 
    Consider the following two observation structures and the resulting PO-OSGs:
    \begin{enumerate}
        \item Each player observes one coordinate.  That is, $\hobs=[3]$ and $\robs=\{A, B\}$ and when $S = (S_1,S_2)$ we have $\hobv = S_1$ and $\robv=S_2$.  We have
        $$
            \E[\actutil(S)\mid\hobv=\obh] = \begin{cases}
                1 & \text{if } \obh = 1, \\
                -\frac{1}{2} & \text{if } \obh = 2, \\
                -4 & \text{if } \obh = 3.
            \end{cases}
        $$
        Hence
        $$
            \polh(\obh) = \begin{cases}
                \on & \text{if } \obh = 1, \\ 
                \off & \text{otherwise}.
            \end{cases}
        $$
        $\rob$'s best response is then $\polr\equiv\wait$.  The expected payoff for $(\polh,\polr)$ is then $1/3$.

        \item $\rob$ has the same observations, but $\hum$ only sees whether $S_1=3$.  Now $\hobs=\{0, 1\}$ and $\hobv = \I(S_1 = 3)$.  Now 
        $$
            \E[\actutil(S)\mid\hobv=\obh] = \begin{cases}
                1/4 & \text{if } \obh = 0, \\
                -4 &\text{if } \obh=1.
            \end{cases}
        $$
        Thus
        $$
            \polh(\obh) = \begin{cases}
                \on & \text{if }\obh = 0, \\
                \off &\text{if }\obh=1.
            \end{cases}
        $$
        $\rob$'s best response is now
        $$
            \polr(\obr) = \begin{cases}
                \wait & \text{if } \obr = A, \\
                \off & \text{if } \obr = B.
            \end{cases}
        $$
        The expected payoff for $(\polh,\polr)$ is now $1/2$.
    \end{enumerate}
    Hence observation structure 2 is better in $\rob$-unaware optimal play than observation structure 1.  Yet structure 1 is strictly more informative for $\hum$ than structure 2.  Clearly structure 1 is weakly more informative for $\hum$ than structure 2.  There is no garbling the other way, as the observations from structure 2 cannot determine the observations in structure 1.
\end{proof}

\subsection{Proof of \Cref{prop:naive-not-better}}
\label{prop:naive-not-better-detailed}

\NaiveNotBetter*

\begin{proof}
    In fact, the previous examples we gave for \Cref{prop:rob-informed,prop:hum-informed} directly work, as $\hum$ already plays the $\rob$-unaware policy in optimal policy pairs.

    \begin{enumerate}[(a)]
        \item Recall the example given in \Cref{prop:rob-informed}. We show the optimal policy pairs in the figure below.
        
        \begin{figure}[ht]
            \begin{center}
                \begin{minipage}{0.21\textwidth}
                        \centering
                        \resizebox{1\textwidth}{!}{\begin{tikzpicture}

\matrix (m) [matrix of nodes,
             nodes={minimum size=1.1cm, draw, anchor=center},
             row sep=-\pgflinewidth,
             column sep=-\pgflinewidth]
{
        \ & $A$ & $B$ \\
        $1.0$ & $+1$ & $-5$ \\ 
        $1.1$ & $-2$ & $+3$ \\ 
        $2.0$ & $+3$ & $+3$ \\
};

\draw (m-1-1.north west) -- (m-1-1.south east);

\node at (m-1-1.north east)  [xshift=-0.3cm, yshift=-0.3cm] {$\rob$};
\node at (m-1-1.south west) [xshift=0.3cm, yshift=0.3cm] {$\hum$};

\node[anchor=east, text=red,xshift=-3pt] (offlabel) at ($(m-2-1.west)!0.5!(m-3-1.west)$) {$\off$};
\node[left=0.1cm of m-4-1, text=gocolor] {$\on$};

\node[above=0.1cm of m-1-2, text=waitcolor] {$\wait$};
\node[above=0.1cm of m-1-3, text=waitcolor] {$\wait$};

\draw[decorate,decoration={brace,mirror,amplitude=5pt},red,thick] 
    (m-2-1.north west) -- (m-3-1.south west);
    
\draw[waitcolor, thick, rounded corners=7pt] 
    ($(m-4-2.north west) - (-0.1,0.1)$) 
    rectangle 
    ($(m-4-3.south east) - (0.1,-0.1)$);

\end{tikzpicture}}
                        \captionof{subfigure}{Expected payoff $=1$}
                    \end{minipage}
                    \quad
                    \begin{minipage}{0.21\textwidth}
                        \centering
                        \resizebox{1\textwidth}{!}{\begin{tikzpicture}

\matrix (m) [matrix of nodes,
             nodes={minimum size=1.1cm, draw, anchor=center},
             row sep=-\pgflinewidth,
             column sep=-\pgflinewidth]
{
        \ & $A$ & $B$ \\
        $1.0$ & $+1$ & $-5$ \\ 
        $1.1$ & $-2$ & $+3$ \\ 
        $2.0$ & $+3$ & $+3$ \\
};

\draw (m-1-1.north west) -- (m-1-1.south east);

\node at (m-1-1.north east)  [xshift=-0.3cm, yshift=-0.3cm] {$\rob$};
\node at (m-1-1.south west) [xshift=0.3cm, yshift=0.3cm] {$\hum$};

\node[left=0.1cm of m-2-1, text=red] {$\off$};
\node[left=0.1cm of m-3-1, text=gocolor] {$\on$};
\node[left=0.1cm of m-4-1, text=gocolor] {$\on$};

\node[above=0.19cm of m-1-2, text=gocolor] {$\go$};
\node[above=0.1cm of m-1-3, text=waitcolor] {$\wait$};

\draw[waitcolor, thick, rounded corners=7pt] 
    ($(m-3-3.north west) - (-0.1,0.1)$) 
    rectangle 
    ($(m-4-3.south east) - (0.1,-0.1)$);

\draw[gocolor, thick, rounded corners=7pt] 
    ($(m-2-2.north west) - (-0.1,0.1)$) 
    rectangle 
    ($(m-4-2.south east) - (0.1,-0.1)$);

\end{tikzpicture}}
                        \captionof{subfigure}{Expected payoff $=\frac{4}{3}$}
                \end{minipage}
            \end{center}
            \caption{The optimal policy pairs in \Cref{ex:hum-informed} when $\hum$ is less informed (left) and when $\hum$ is more informed (right). In OPPs, $\hum$ becoming more informed makes $\rob$ wait strictly less often. These are also $\rob$-unaware OPPs.}
            \label{fig:hum-informed-matrices}
        \end{figure}
        
        In the less informative case, $\hum$'s policy in the optimal policy pair is:
        $$\polh(\obh) = 
            \begin{cases}
                \on &\text{if } \obh = 2.x \\
                \off &\text{if } \obh = 1.x
            \end{cases}
        $$

        This is also the $\rob$-unaware policy, as $$\E[\actutil(S)\mid \hobv=1.x] = -3/4 < 0$$ and $$
        \E[\actutil(S)\mid \hobv=2.x] = +3 > 0.
        $$

        In the more informative case, $\hum$'s policy in the optimal policy pair is:
        $$\polh(\obh) = 
            \begin{cases}
                \on &\text{if } \obh \in \{1.1, 2.0\} \\
                \off &\text{if } \obh = 1.0
            \end{cases}
        $$

        This is also the $\rob$-unaware policy, as we have the following three results: $$\E[\actutil(S)\mid \hobv=1.0] = -3 < 0, $$ $$\E[\actutil(S)\mid \hobv=1.1] = +1/2 > 0, $$and $$
        \E[\actutil(S)\mid \hobv=2.x] = +3 > 0.
        $$

        Hence the unique optimal policy pair is also the unique $\rob$-unaware optimal policy pair in both cases. 

        \item Recall that in both cases, $\hum$ observes the row in the following table which shows how $\actutil$ depends on the state:
        \begin{table}[ht]
            \centering
            \begin{tabular}{c|c|c|c}
              \backslashbox{$\hum$}{$\rob$} & $F$ & $S$ & $M$\\ \hline
            $N$ & $+2$ & $+3$ & $+4$ \\ \hline
            $E$ & $-4$ & $-1$ & $+2$ 
            \end{tabular}
        \end{table}

        Therefore, the $\rob$-unaware human policy is:

        $$\polh(\obh) = 
            \begin{cases}
                \on &\text{if } \obh = N \\
                \off &\text{if } \obh = E
            \end{cases}
        $$

        as $$\E[\actutil(S)\mid \hobv= N] = +3 > 0, $$ and $$
        \E[\actutil(S)\mid \hobv=E] = -1 < 0.
        $$

        This is identical to the human policy of the optimal policy pair of both cases of the example in \Cref{prop:rob-informed}. Hence the unique optimal policy pair is also the unique $\rob$-unaware optimal policy pair in both cases.
    \end{enumerate}
\end{proof}

\section{The Complexity of Solving PO-OSGs}
\label{app:complexity}

Computing optimal policy pairs in off-switch games without partial observability is easy. $\rob$ can simply compute the expected value of each action and play the highest one, $\hum$ can compute the expected value of ON and OFF then do the same.

With the introduction of partial observability, the landscape becomes much more interesting.  \citet{bernstein2002dec} showed that for decentralized POMDPs, of which PO-OSGs are instances of, deciding whether a policy pair exists with utility above a given threshold is NEXP-complete.
Given their specialized nature, finding optimal policy pairs in PO-OSGs is easier, but still computationally difficult.

\begin{theorem}
    \label{thm:posg-np-complete}
    The following decision problem is NP-Complete: given a PO-OSG and a natural number $k$, decide if there exists a policy pair $(\polh, \polr)$ with expected payoff at least $k$.
\end{theorem}
\begin{proof}
    By \Cref{cor:uniquedopp}, we may consider only deterministic policy pairs.  That is, if there is a policy pair $(\polh,\polr)$ with expected payoff at least $k$, then there is also a deterministic optimal policy pair with expected payoff at least $k$.

    To show that our decision problem is in NP, note that given an optimal policy pair to determine if the optimal policy pair has expected payoff bigger than $k$, it suffices to compute a linear combination of payoffs: iterating through each pair of human-assistant observations, using the policy to find expected payoff in constant time, and scaling by the probability of those observations. This gives us a $\bO{\abs{\robs}\cdot \abs{\hobs}}$ time algorithm for verifying a solution.

    To show it is NP-hard, we provide a reduction from MAXCUT (which is known to be NP-complete). Consider the following problem: given a graph $G = (V, E)$ and value $k$, decide if there exists a cut of size at least $k$. Let $n = \abs{V}$. We can construct the following equivalent PO-OSG.
    The state space consists of pairs of vertices, $\states = V \times V$. The human can see the first vertex, $\hobs = V$, the assistant the second $\robs = V$. Each pair of vertices is equally likely.  Clearly this game can be constructed in polynomial time.

    The utility of acting in state $(v_1, v_2) \in \states$,
    \begin{equation*}
        \actutil((v_1, v_2)) = \begin{cases}
            -n^4 & \text{if $v_1 = v_2$}, \\
            n^2 & \text{if $(v_1, v_2) \in E$}, \\
            0 & \text{otherwise},
        \end{cases}
    \end{equation*}
    and $\offutil\equiv 0$. Hence the players try to act exactly when they receive adjacent vertices and never when they have the same vertex. This setup encourages them to choose a cut and only act when they see a vertex in their part.

    If a cut $(V^\rob, V^\hum)$ of size $k$ exists in $G$, then there exists a policy pair with expected payoff at least $k$. Indeed, $\rob$ can play $\wait$ when $v_1 \in V^\rob$, and $\off$ otherwise. $\hum$ responds by playing $\on$ when $v_2 \in V^\hum$ and $\off$ otherwise. Formally:
    \begin{align*}
         \polr(\obr) = \begin{cases}
            w(a) & \text{if } \obr\in V^\rob, \\
            \off & \text{if } \obr\in V^\hum, 
        \end{cases} \\
        \polh(\obh) = \begin{cases}
            \on & \text{if } \obh\in V^\hum, \\
            \off & \text{if } \obh\in V^\rob.
        \end{cases}
    \end{align*}

    When $\hum$ and $\rob$ coordinate on playing $a$, they must have the following expected utility:
    \begin{align*}
        &\frac{1}{n^2}\sum_{(\obh,\obr)\in V^\hum\times V^\rob} \actutil((\obh,\obr)) \\
        &= \sum_{(\obh,\obr)\in V^\hum\times V^\rob} \I((\obh,\obr)\in E) \geq k.
    \end{align*}

    In the other direction, suppose that $(\polr,\polh)$ is a deterministic policy pair achieving expected payoff at least $k$. We will show that there exists a cut of size $k$.

    First, notice that there is never an incentive for $\rob$ to play $a$. The expected utility, regardless of $\hum$'s observation, is always at most:
    \begin{align*}
        &\frac{1}{n^2}\sum_{(\obh,\obr)\in V^\hum\times V^\rob} \actutil((\obh,\obr)) \\
        &\leq \frac{1}{n^2}\left(\frac{n(n - 1)}{2} \cdot n^2 - n^4\right) < 0.
    \end{align*}
    The cost of both vertices being the same is simply too high for $\rob$ to risk playing $a$.  Moreover, for this reason, there is no $v\in V$ such that $\polr(v) = w(a)$ and $\polh(v) = \on$.

    This allows us to define the following disjoint sets of vertices:
    \begin{align*}
        &V^\hum = \{v\in V: \polh(v) = \on\}, \\
        &V^\rob = \{v\in V: \polr(v) = w(a)\}, \\
        &V^0 = V\setminus(V^\hum\cup V^\rob).
    \end{align*}

    Let $V_1 = V^\hum$ and $V_2 = V^\rob\cup V^0$.  Consider the cut $(V_1, V_2)$. The size of this cut must be:
    \begin{align*}
        &\sum_{(v_1, v_2)\in V_1\times V_2} \I((v_1, v_2)\in E) \\
        &\geq \sum_{(v_1, v_2)\in V^\hum\times V^\rob} \I((v_1, v_2)\in E) \\
        &= \frac{1}{n^2} \sum_{(v_1, v_2)\in V^\hum\times V^\rob} n^2\I((v_1, v_2)\in E).
    \end{align*}
    We can rewrite this to iterate through all pairs of vectors with the following indicator:
    \begin{align*}
        \frac{1}{n^2} \sum_{(v_1, v_2)\in V\times V} \I(\polh(v_1) = &\on \land \polr(v_2) = w(a))\\
        &\cdot  n^2\I((v_1, v_2)\in E).
    \end{align*}
    Because $\rob$ never plays $a$, this is the expression of the expected utility of $(\polr, \polh)$, and so is at least $k$. Thus, the max cut is of size at least $k$, proving that a policy of utility at least $k$ exists if and only if a cut of size $k$ exists, as claimed!
\end{proof}

By comparison, computing $\rob$-unaware optimal policy pairs (assuming constant-time lookups) is easy. Consider the following two-step algorithm:
\begin{enumerate}
    \item Compute $\polh$ in $\mathcal{O}(\text{poly}(|\states|,|\hobs|,|\robs|))$ time.  For $\obh\in\hobs$:
    \begin{enumerate}
        \item Set $\Delta = \E[\actutil(S)-\offutil(S)\mid\hobv=\obh]$, which we can calculate in $\mathcal{O}(\text{poly}(|\states|,|\robs|))$ via Bayes' rule and LOTP.

        \item Set $\polh(\obh)$ to $\on$ if $\Delta \geq 0$ and $\off$ otherwise.
    \end{enumerate}
    \item Compute $\polr$ in $\mathcal{O}(\text{poly}(|\states|,|\hobs|,|\robs|))$ time.  For $\obr\in\robs$:
    \begin{enumerate}
        \item Set
        \begin{align*}
            \Delta_a &= \E[\actutil(S)\I(\polh(\hobv)=\on)\mid\robv=\obr] \\
                     &\quad + \E[\offutil(S)\I(\polh(\hobv)=\off)\mid\robv=\obr] \\ 
                     &\quad - \E[\actutil(S)\mid\robv=\obr]
        \end{align*}
        and 
        \begin{align*}
            \Delta_\off &= \E[\actutil(S)\I(\polh(\hobv)=\on)\mid\robv=\obr] \\
                        &\quad + \E[\offutil(S)\I(\polh(\hobv)=\off)\mid\robv=\obr] \\
                        &\quad - \E[\offutil(S)\mid\robv=\obr].
        \end{align*}
        We can calculate these in $\mathcal{O}(\text{poly}(|\states|,|\hobs|,|\robs|))$ time as before.

        \item Now set
        $$
            \polr(\obr) = \begin{cases}
                a &\text{if } \Delta_a < 0, \\
                \off & \text{if } \Delta_\off < 0, \\
                \wait & \text{otherwise}.
            \end{cases}
        $$
    \end{enumerate}
\end{enumerate}
This algorithm calculates the $\rob$-unaware optimal policy pair in $\mathcal{O}(\text{poly}(|\states|,|\hobs|,|\robs|))$ time, as claimed.

The results of this section vindicate our choice to study $\rob$-unaware optimal policy pairs. $\rob$-unaware optimal policy pairs are significantly easier to calculate in general than optimal policy pairs. 

\section{PO-OSGs as Assistance Games}\label{app:posgs-poags}

Partially-Observable Off-Switch Games (PO-OSGs) are special cases of assistance games. Recall that we formally define PO-OSGs in \Cref{df:bosg}. \citet{emmons2024tampering} define partially observable assistance games (POAGs) by the following tuple (with minor notational modifications for ease of comparison with our PO-OSG definition): 
\[(\states, \{\haxns, \raxns\},T,\{\Theta,\util\}, \{\hobs,\robs\},\obsdist,\statedist,\gamma)\]
$\states$ is a set of states, $\haxns$ and $\raxns$ are human and assistant action sets, $T: \states \times \haxns \times \raxns \to \Delta(\states)$ is a transition function, $\Theta$ is a set of utility parameters describing the human's possible preferences, $\util: \states \times \haxns \times \raxns \times \Theta \to \R$ is a shared utility function, $\hobs$ and $\robs$ are human and assistant observation sets, $\obsdist: \states \times \haxns \times \raxns \to \Delta(\hobs \times \robs)$ is a conditional observation distribution, $\statedist \in \Delta(\states \times \Theta)$ is an initial distribution over states and utility parameters, and $\gamma \in [0,1]$ is a discount factor.

We can present a PO-OSG $(\states, (\hobs, \robs, \obsdist), \statedist, \util)$ as a POAG instead. In PO-OSGs, we roll the human's preference parameters $\Theta$ into $\states$ and $\hobs$ to capture the fact that the human knows her own preferences but the assistant may not. So the corresponding POAG has of states $\states_2$, human observations $\hobs_2$, and preference parameters $\Theta$ such that $\states = \states_2 \times \Theta$ and $\hobs = \hobs_2 \times \Theta$. $\raxns$ and $\robs$ stay the same in the PO-OSG and POAG presentations of the game, with $\raxns=\{a, w(a), \text{OFF}\}$. In PO-OSGs without communication, the transition function $T$ is unimportant, as there is only one time step in the game. With communication, $T$ intuitively induces a transition such that the new state allows both agents to observe the other agent's message. $\util$ is the same in the PO-OSG and the POAG, except it does not depend on $\Theta$ in the PO-OSG because $\Theta$ is rolled into $\states$. $\obsdist$ and $\statedist$ are the same, with minor modifications to account for the fact that we rolled $\Theta$ into $\states$ in the PO-OSG. Finally, $\gamma$ is irrelevant when there is no communication, and $\gamma=1$ when there is communication to ensure there is no discounting.

Some of the generalizations of PO-OSGs that we describe as future work in \Cref{sec:conclusion}, such as incorporating longer sequences of interactions, can likely be supported within the POAG framework as well.

\end{document}